\documentclass[12pt]{article}
\usepackage[titletoc,title]{appendix}
\usepackage[margin=1in]{geometry}
\usepackage{setspace}
\usepackage{eurosym}
\usepackage{textcomp}
\usepackage[hidelinks]{hyperref}
\usepackage{color}
\usepackage{microtype}
\usepackage[english]{babel}
\usepackage{authblk}
\usepackage{amsmath, mathtools, scalerel}
\usepackage{amsfonts}
\usepackage{amssymb}
\usepackage{amsthm}
\usepackage{mathrsfs}
\usepackage{lscape}
\usepackage{graphicx}
\usepackage{breqn}
\usepackage{sgamevar}
\usepackage{tikz}
\usepackage{pgfplots}
\usepackage{float}
\usepackage[round]{natbib}
\usepackage{fancyhdr}
\usepackage{authblk}
\usepackage{multirow,array}
\usepackage[flushleft]{threeparttable}
\usepackage{bm}
\usepackage{dsfont}
\usepackage[multiple]{footmisc}
\usetikzlibrary{decorations.pathreplacing,angles,quotes}
\allowdisplaybreaks

\newtheorem{theorem}{Theorem}
\newtheorem{proposition}{Proposition}
\newtheorem{corollary}{Corollary}
\newtheorem{lemma}{Lemma}

\newtheorem{claim}{Claim}

\theoremstyle{remark}

\DeclareMathOperator*{\argmax}{arg\,max}

\newcommand{\R}{\mathbb{R}}
\newcommand{\N}{\mathscr{N}}
\newcommand{\wrt}[1]{\mathrm{d}{#1}}

\title{\textsc{Unraveling Coordination Problems}\thanks{For valuable comments and suggestions, I thank Eric van Damme, Tabar\'{e} Capit\'{a}n, Reyer Gerlagh, Olga Heijmans-Kuryatnikova, Eirik Gaard Kristiansen, J\'{o}zsef S\'{a}kovics, Robert Schmidt, Christoph Schottm\"{u}ller, Florian Schuett, Sigrid Suetens, and Florian Wagener. }}
\author{Roweno J.R.K. Heijmans\thanks{Department of Business and Management Science, NHH Norwegian School of Economics, Helleveien 30, 5045 Bergen, Norway. Email: roweno.heijmans@nhh.no.}}

\date{\today}

\begin{document}
\maketitle

\begin{abstract}
Strategic uncertainty complicates policy design in coordination games. To rein in strategic uncertainty, the Planner in this paper connects the problem of policy design to that of equilibrium selection. We characterize the subsidy scheme that induces coordination on a given outcome of the game as its unique equilibrium. Optimal subsidies are unique, symmetric for identical players, continuous functions of model parameters, and do not make the targeted strategies strictly dominant for any one player; these properties differ starkly from canonical results in the literature. Uncertainty about payoffs impels policy moderation as overly aggressive intervention might itself induce coordination failure.


\end{abstract}

\section{Introduction}
In coordination problems, players face strategic uncertainty that forces them to second-guess the strategies of their opponents. Pessimistic beliefs can become self-fulfilling and lead to coordination failure. A worthwhile project may not take off simply because investors believe others will not invest. A promising network technology may never mature only because potential adopters are pessimistic about adoption by others. An infectious disease may not get eradicated solely on the ground that governments believe other nations will not attempt to. The possibility of costly coordination failures motivates intervention.

The usual rationale for policy intervention is to correct market failures introduced by externalities. All market failures are not equal, however, and it is crucial for policy design to know the type of externality an intervention targets. One kind of externality arises when there exists a gap between the private and social value of behavior. Thus, an individual household's greenhouse gas emissions may be higher than socially optimal as it ignores the effects its emissions have on others. Such externalities can be addressed though Pigouvian taxes or subsidies. Another, more complicated kind of externality arises in coordination problems where individual actions are strategic complements \citep{bulow1985multimarket}. Strategic complementarity results in multiple, Pareto-ranked equilibria and opens the door to coordination failures. Thus, a renewable technology may well provide a viable replacement for fossil fuels but only if sufficient capacity is installed; there hence are multiple equilibria, and renewables might never mature despite their potential (and known) advantages \citep{barrett2006climate}. Such externalities cannot be solved through simple Pigouvian policy. The goal of this paper is to design optimal policies for coordination problems. To streamline the narrative, we focus on subsidies.

The main results in this paper characterize the subsidy scheme that induces a given outcome of a coordination game as its unique equilibrium. This subsidy scheme is unique. Moreover, subsidies pursuant to the scheme are (i) symmetric for identical players; (ii) continuous functions of model parameters; and (iii) do not make the targeted strategies strictly dominant for any of the players. These properties run counter to several important and well-known results in the literature \citep[\textit{cf.}][]{segal2003coordination,winter2004incentives,bernstein2012contracting,sakovics2012matters,halac2020raising}. Two features of the problem considered here cut to the core of these opposing results. 

First, this paper deals with settings in which the Planner is uncertain about the efficient outcome of the game at the time she offers her subsidies. She might, for example, want to subsidize one out of multiple competing technologies yet lack the knowledge of their learning curves necessary to make a well-informed decision \citep[see][for an anlysis of this exact problem]{cowan1991tortoises}.\footnote{The historical records are replete with examples of policymakers who faced uncertainty about the efficient course of action -- and chose wrongly. \cite{cowan1990nuclear} describes the history of nuclear power generation. Nowadays, light water nuclear reactors are the dominant technology. This situation can be traced back to Captain Hyman Rickover of the U.S.\ Navy, whose preference for light water drove the early development of this technology led to its eventual domination of the field. There now is compelling evidence that two competing technologies, both of which were known to Captain Rickover, are economically and technologically superior to light water nuclear reactors. Similarly, \cite{cowan1996sprayed} discuss competing pest control strategies in agriculture. They show that today's heavy reliance on pesticides -- a consequence of targeted policies in the 1930s and 1940s -- is inefficient. Evidence indicates that a competing technology that already existed at the time, Integrated Pest Management, is technologically and economically superior to pesticides. This wasn't known, however, when policymakers first had to choose which type of pest control to pursue.}  Uncertainty of this kind impels policy moderation as overly aggressive intervention might itself become a source of, rather than solution to, coordination failure. While the problem of policy design in coordination games under (endogenous) uncertainty is well-studied and -understood \citep[\textit{cf.}][]{angeletos2006signaling,sakovics2012matters,halac2021rank,halac2022addressing,kets2021theory,kets2022value}, our focus on uncertainty about the efficient outcome of the game sets this paper apart from earlier contributions.

Second, the Planner in this paper connects the problem of policy design to that of equilibrium selection. By pinning down precisely how the game will be played, equilibrium selection allows the Planner to design her policy in response to players' actual, rather than hypothetical, strategic beliefs. This  reining in of strategic uncertainty implies she need not make the strategies she wants players to pursue strictly dominant for any one of them as, in the unique equilibrium selected, no player has reason to believe that others will not play the targeted strategy.\footnote{\cite{sakovics2012matters} also connect the problems of policy design and equilibrium selection. In \cite{sakovics2012matters}, however, the action that is subsidized depends upon the efficient outcome of the game.}

Let us make the discussion a bit more precise. The model in this paper consists of a Planner and $N$ (heterogeneous) players each of whom independently chooses an action from a binary set $\{0,1\}$. If player $i$ plays 0, his payoff is $c_i$. When instead player $i$ plays 1, his payoff is the sum of two components. The first component, $x$, is a state of Nature. The second component, $w_i:\{0,1\}^{N-1}\to\R$, gives the externalities that other players' actions impose upon him. The analysis centers around coordination games, or games with strategic complementarities, in which $w_i$ is increasing in the number of players that play 1. A Planner publicly announces subsidies to players who play 1. The problem of the Planner is to find the vector of subsidies $\tilde{s}=(\tilde{s}_i)$ that induces coordination on $(1,1,...,1)$ for all $x>\tilde{x}$, where the critical state $\tilde{x}\in\R$ is chosen by the Planner. The paper also explores a number of extensions and special cases of the base model, including: principal-agent models; games of regime change; asymmetric policy targets; and games with heterogeneous externalities.

Were a player informed about the actions of his opponents, his problem would be trivial. Yet players do not typically possess such information. In a coordination game with multiple Nash equilibria, the resulting strategic uncertainty forces players to second-guess the actions and beliefs of others. This complicates the Planner's problem: even if a policy makes coordination on $(1,1,...,1)$ \textit{an} equilibrium for all $x>\tilde{x}$, there may yet be others. Unless the Planner can coordinate play on her most-preferred equilibrium -- a power economists have been reluctant to grant \citep[\textit{cf.}][]{segal1999contracting,segal2003coordination,winter2004incentives,sakovics2012matters,bernstein2012contracting,halac2020raising,halac2021rank} -- the purpose of her policy is not simply to make the targeted outcome an equilibrium. Instead, she seeks to attract coordination on one, rather than another, equilibrium. She therefore cannot separate the issue of policy design from that of equilibrium selection.

The Planner in this paper deals with equilibrium selection using a global games approach. Pioneered by \cite{carlsson1993global}, global games are incomplete information games in which players do not observe the true game they play but only a private and noisy signal of it. In our game, players do not know the hidden state $x$; rather, each player $i$ observes a private noisy signal $x_i^\varepsilon$ of $x$. We focus on regulatory environments in which the Planner does not know $x$ either or else, if she does, must commit to her policy before Nature draws $x$; hence, the Planner's choice of policy cannot signal any private knowledge she might possess \citep{angeletos2006signaling}.\footnote{More specifically, the problem of the Planner is not one of Bayesian persuasion or information design \citep[cf.][]{kamenica2011bayesian,bergemann2016information,ely2017beeps,mathevet2020information}.} Given this information structure, it is impossible to tackle the Planner's problem directly. Instead, the analysis first solves a slightly modified version of her problem: find that subsidy scheme $\tilde{s}$ subject to which the unique equilibrium strategy of each player $i$ is to choose 1 whenever his signal $x_i^\varepsilon$ exceeds $\tilde{x}$. In the limit as signals become arbitrarily precise, this implies coordination on $(1,1,...,1)$ for all $x>\tilde{x}$ with probability 1 and thus solves the Planner's original problem as well. Our main result shows that the subsidy scheme $\tilde{s}$ exists, that it is unique, and provides a characterization.

An interesting economic consequence of equilibrium selection in the global game is that even ``small'' subsidies exhibit clear equilibrium effects. A subsidy raises player $i$'s incentive to play 1. Because players in a coordination game want to match actions, the subsidy to player $i$ also (indirectly) increases player $j$'s incentive to play 1. This, in turn, makes the playing 1 even more attractive for player $i$, and so on. If subsidies are common knowledge, we obtain an infinitely compounded feedback loop of policy; see Figure~\ref{fig_indirect} for an illustration. Because of this, to induce coordination on a given strategy vector the Planner need not make the associated strategies strictly dominant for any of the players.

\begin{figure}[h]
	\centering
	\includegraphics[width=0.4\textwidth, angle=0, clip=true, trim=1cm 0.7cm 6cm 0.9cm]{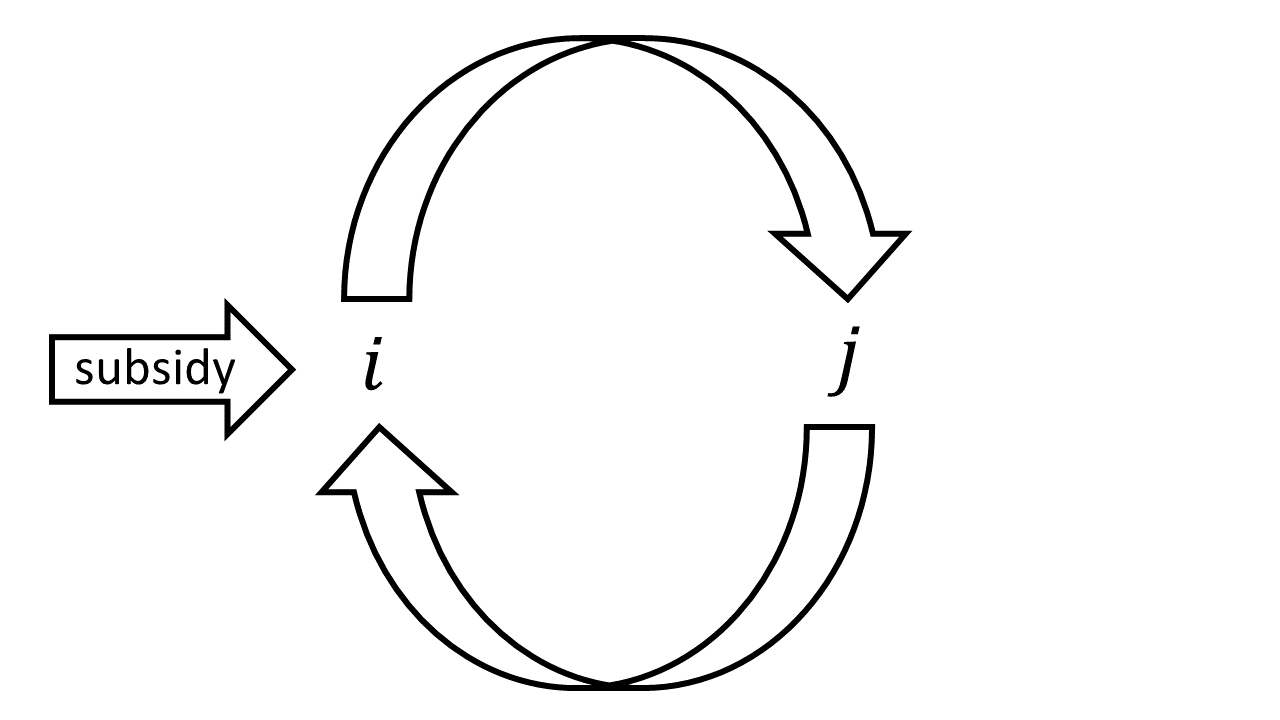}
	\caption{In coordination game, a subsidy kickstarts an infinitely compounded positive feedback loop on players' incentives to play the subsidized action: the unraveling effect.}
	\label{fig_indirect}
\end{figure}

\vspace{3mm}
\textit{Related literature.}---A closely related paper is \cite{sakovics2012matters}, who study policy design in a global game of regime change. Games of regime change are coordination games in which a status quo is abandoned, causing a discrete change in payoffs, once a sufficiently large number of agents take an action against it. \cite{sakovics2012matters} find that an optimal policy fully subsidizes a subset of players, targeting those who matter most for regime change and/or have least incentive to take an action against the regime. These results provide a stark counterpoint to the findings in this paper, which say that an optimal policy subsidizes \textit{all} players partially. The difference is a consequence of the distinct information structures considered. In \cite{sakovics2012matters}, the (ex post) efficient outcome of the game is known to the Planner when she offers her policy; in this paper, it is not. The same distinction also set this paper apart from the broader literature on policy design in global games \citep{goldstein2005demand,angeletos2006signaling,angeletos2007dynamic,sakovics2012matters,edmond2013information,basak2020diffusing}.

Another related paper is \cite{halac2020raising}, who study the problem of a firm that seeks to raise capital from multiple investors to fund a project. The project succeeds only if the capital raised exceeds a stochastic threshold; the firm offers payments contingent on project success. \cite{halac2020raising} identify conditions under which larger investors receive higher per-dollar returns on investment in an optimal policy, thus perpetuating inequalities. The focus on contingent per-dollar returns in \cite{halac2020raising} is different from the approach in this paper, in which actions are binary and subsidies are paid regardless of eventual outcomes.

This paper is also related to the literature on principal-agent contracting, see \cite{winter2004incentives} and \cite{halac2021rank} in particular. Contrasting sharply with the findings presented here, the seminal result in \cite{winter2004incentives} is that optimal mechanisms are inherently discriminatory under complete information -- no two agents are rewarded equally even when agents are symmetric. \cite{halac2021rank} extend the model in Winter (2004) to allow for asymmetries among the agents and private contract offers and find that symmetric agents are offered identical rewards in an optimal contract. Like \cite{halac2021rank}, this paper finds that an optimal policy treats symmetric players identically. Interestingly, however, the results in \cite{halac2021rank} depend critically upon contract offers being private; in contrast, it is crucial that offers are common knowledge for the results in this paper.

Another literature to which this paper connects is that on contracting with externalities \citep[e.g.,][]{segal1999contracting,segal2003coordination,segal2000naked,bernstein2012contracting}. \cite{segal2003coordination} and \cite{bernstein2012contracting} consider complete information contracting problems that, save for the informational environment, are essentially equivalent to the game studied in this paper. They establish optimality of the \textit{divide and conquer} mechanism in which the Planner first ranks all players; given the ranking, each player is offered a subsidy that incentivizes him to play the subsidized action assuming all players who precede him in the ranking also play this action while those after him do not. \cite{bernstein2012contracting} derive the optimal ranking of players in such a policy. Like the mechanism derived in \cite{winter2004incentives}, an (optimal) divide and conquer scheme is inherently discriminatory and treats symmetric agents asymmetrically.

Some authors study coordination games using solution concepts other than Nash equilibrium. A notable example is \cite{kets2022value}, who consider policy design in (symmetric $2\times 2$) coordination games using the concept of \textit{introspective equilibrium} developed by \cite{kets2021theory}. \cite{kets2022value} find that subsidies have both direct and indirect effects in coordinatiom games; in contrast to the results in this paper, however, the direct and indirect effects in an introspective equilibrium can affect incentives in opposite directions. Moreover, like the applied global games literature, \cite{kets2022value} focus on games in which the efficient outcome of the game -- and thus the outcome to be subsidized -- is known a priori.

The remainder of the paper is organized as follows. Section~\ref{sec_example} introduces a simple example to develop a basic intuition for why policy design cannot be separated from equilibrium selection in a coordination game. Section~\ref{sec_model} introduces the model and the concepts needed for the analysis. Section~\ref{sec_subsidies} introduces the Planner's problem and states out main result. Section~\ref{sec_main} presents the core of the analysis. Various special cases and extensions of our model are discussed on Section~\ref{sec_special}. Section~\ref{section_concl} discusses and concludes. All proofs are in the Appendix.

\section{A Simple Example}\label{sec_example}
This section develops an intuition for the main results in this paper in a highly simplified example. In particular, it illustrates why connecting the issue of policy design and equilibrium selection is important in coordination games.

There are two players who can participate in a project. The cost of participation to player $i$ is $c_i>0$. If the project succeeds, player $i$ earns a payoff $b_i>c_i$. The project succeeds if and only if both players participate. The payoff to not participating, the outside option, is $-x$.

A Planner publicly offers each participating player $i$ a subsidy $s_i$. The Planner's problem is to find that subsidy scheme which induces players to coordinate on joint participation for all $x>0$ as the unique equilibrium of the game.

Suppose first that we were to approach the Planner's problem without taking care of equilibrium selection. Since this yields a coordination problem with multiple strict Nash equilibria, the Planner now operates under the assumption that players can hold essentially \textit{any} strategic beliefs. In particular, letting $q_j\in[0,1]$ denote the probability that player $i$ attaches to his opponent $j$ participating, both $q_j=1$ and $q_j=0$ are supported as consistent equilibrium beliefs in coordination games with multiple equilibria. It is therefore not difficult to see that a policy \textit{guarantees} participation by both players only if it makes participation strictly dominant for at least one of them. Moreover, making participation dominant for only one player is also sufficient to guarantee project success since the other player, taking participation by the subsidized player as given, will also participate. The Planner should therefore offer a subsidy $s_i\geq c_i$ to only one player $i$.

This canonical result breaks down once we connect the problem of policy design to that of equilibrium selection. More precisely, the necessity of subsidizing at least one player all the way toward strict dominance derives from mutual non-participation always being an equilibrium of the game, justifying players' beliefs that $q_i=q_j=0$. By introducing uncertainty and turning the problem into a global game \citep{carlsson1993global}, we can -- for any subsidy scheme $(s_i,s_j)$ -- select a unique equilibrium of the underlying coordination problem. Equilibrium uniqueness places severe restrictions upon players' beliefs $q_i$ and $q_j$; restrictions that, in a coordination game, turn out crucial for policy design.

Our main anlysis (in Section~\ref{sec_main}) will imply that in the unique equilibrium selected in this simple game a rational player $i$ participates for all $x$ at which participation is a strict best response to all (hypothetical) beliefs $q_j\geq 1/2$. Observe that, for generic $q_j$ and given a subsidy $s_i$, participation yields player $i$ an expected payoff of $q_j\cdot [b_i+s_i-c_i]+(1-q_j)\cdot [s_i-c_i]$ whilst non-participation pays him $-x$. The former is clearly increasing in $q_j$ so that to solve for an optimal subsidy, it suffices to consider only the (possibly wrong) belief $q_j=1/2$. Hence, the optimal subsidy that induces player $i$ to participate for all $x>0$ is given by
\[s_i=c_i-\frac{b_i}{2},\]
for each player $i$. Observe that, in the global game, \textit{both} players are offered subsidies neither of which makes participation strictly dominant for all $x>0$.

We note that an optimal policy does two things at once. First, the subsidy $s_i$ ensures that participation is a (strict) best response to player $i$'s belief $q_j\geq1/2$ for all $x>0$. Second, and because $s_i$ does this, it allows player $j$ to disregard the belief $q_j=0$ that (in the coordination problem prior to equilibrium selection) implied the necessity of subsidizing one player to strict dominance.

This discussion serves as a simple illustration of some key results and comparisons that we present in the sections to follow. Our main result, Theorem~\ref{thm_subsidy}, characterizes the optimal subsidy scheme for more general coordination games when we approach policy design in the context of equilibrium selection. We also describe the process of equilibrium selection more explicitly to show exactly how the disciplining of players' strategic beliefs comes about, justifying the restrictions simply imposed in this illustrative example.

\section{The Game}\label{sec_model}
Consider a normal form game played by players in a set $\N=\{1,2,...,N\}$, indexed $i$, who simultaneously choose binary actions $a_i\in\{0,1\}$. Define $a_{-i}:=a\setminus\{a_i\}$, $\overline{a}:=(1,1,...,1)$, $\underline{a}:=(0,0,...,0)$, $\overline{a}_{-i}:=\overline{a}\setminus\{a_i\}$, and $\underline{a}_{-i}:=\underline{a}\setminus\{a_i\}$. When $a$ is played, player $i$ who chooses 1 in $a$ gets payoff $x+w_i(a_{-i})$; when instead player $i$ chooses $0$ in $a$, his payoff is $c_i$. Here, $w_i(a_{-i})$ describes the externalities on player $i$ deriving from other players' actions. To simplify the exposition, the main analysis assumes that $w_i$ depends upon $a_{-i}$ only through the aggregate action and we will often write $w_i(\sum_{j\neq i}a_j)$; Section~\ref{sec_heterogeneous} explores generalizations of the game in which externalities depend upon the exact subset of players who play 1. The variable $x$ is a hidden state of Nature. Lastly, $c_i$ is player $i$'s payoff to playing 0, which in some interpretations of the model is best thought of as the cost of playing 1. Combining these elements, the payoff to player $i$ is given by
\begin{equation}\label{eq_payoff}
	\pi_i(a\mid x)=
	\begin{cases}
		x+w_i\left(\sum_{j\neq i}a_j\right)\quad&\text{if}\quad a_i=1\text{ in } a,\\
		c_i\quad&\text{if}\quad a_i=0\text{ in } a.
	\end{cases}
\end{equation}

We restrict attention to games with strategic complementarities meaning that $w_i(\sum_{j\neq i}a_j)$ is increasing in $\sum_{j\neq i}a_j$, i.e.\ $w_i(n+1)\geq w_i(n)$ for all $n=0,...,N-2$. In the canonical example of a joint investment problem, the action $a_i=1$ is interpreted as investment and $c_i$ as the cost of investing \citep{sakovics2012matters}. Alternatively, actions might represent the choice to use of a particular kind of network technology and $c_i$ is the cost differential between technologies \citep{cowan1991tortoises,bjorkegren2019adoption,leister2022social}. Or actions could describe the decisions to work or shirk by agents working on a common project such that $c_i$ is agent $i$'s cost of effort and $w_i$ his (discrete) benefit from project success, see \cite{winter2004incentives} and \cite{halac2021rank}.

The above elements combined describe a game of complete information $\Gamma(x)$. In $\Gamma(x)$, we define a player's \textit{incentive} to choose 1 as the gain from playing 1, rather than 0, or
\begin{equation}\label{eq_incentive}
	u_i(a_{-i}\mid x)=\pi_i(1,a_{-i}\mid x)-\pi_i(0,a_{-i}\mid x)=x+w_i\left(\sum_{j\neq i}a_j\right)-c_i.
\end{equation}

Observe that, given $a_{-i}$, a player's incentive $u_i$ to play 1 is strictly increasing in $x$. Denote $x^0_i:=c_i-w_i(0)$ and $x^N_i:=c_i-w_i(N-1)$. One has $u_i(\overline{a}_{-i}\mid x^0_i)=u_i(\underline{a}_{-i}\mid x^N_i)=0$. In other words, to each player $i$ playing 1 is strictly dominant for all $x>\overline{x}^0_i$; playing 0 is strictly dominant for $x<\underline{x}_i^N$. Define ${x}^N:=\max\{{x}_i^N\mid i\in\N\}$, ${x}^0:=\min\{{x}^0_i\mid i\in\N\}$, $\underline{x}=\min\{{x}^0_i\mid i\in\N\}$, and  $\overline{x}=\max\{{x}_i^N\mid i\in\N\}$. Let $[\underline{x},\overline{x}]$ be nonempty so that, for all $x$ in $[\underline{x},\overline{x}]$, $\Gamma(x)$ is a true coordination game with multiple strict Nash equilibria. 

To reflect the many uncertainties that exist in the real world, we assume that the state of nature $x$ is hidden. Instead, it is common knowledge among the players that $x$ is drawn from a continuous prior density $g:\mathcal{X}\to\R$ and that each player $i$ receives a private noisy signal $x_i^{\varepsilon}$ of $x$, given by
\begin{equation}\label{eq_signal}
	x_i^\varepsilon=x+\varepsilon\cdot \eta_i,
\end{equation}
where $\mathcal{X}=[\underline{X},\overline{X}]\supseteq[\underline{x}-\varepsilon,\overline{x}+\varepsilon]$ is closed. One can think of $x_i^\varepsilon$ as the player's \textit{type}. The random variable $\eta_i$ is a noise term that is distributed i.i.d.\ on $[-1/2,1/2]$ according to a continuously differentiable distribution $F$, and $\varepsilon>0$ is a scaling factor.\footnote{The assumption that the support of $\eta_i$ is $[-1/2,1/2]$ is without loss. If $\eta_i$ were systematically biased, rational players would simply take that into account when forming their posteriors. Moreover, we could also allows the noise distribution to have support on the entire real line without great technical complications.} We write $\Gamma^\varepsilon$ for the game of incomplete information about $x$. 

Let $x^{\varepsilon}=(x_i^{\varepsilon})$ denote the vector of signals received by all players, and let $x^{\varepsilon}_{-i}$ denote the vector of signals received by all players but $i$, i.e. $x^{\varepsilon}_{-i}=(x_j^{\varepsilon})_{j\neq i}$. Note that player $i$ observes $x_i^{\varepsilon}$ but neither $x$ nor $x^{\varepsilon}_{-i}$. We write $F_i^\varepsilon(x,x_{-i}^\varepsilon\mid x_i^\varepsilon)$ for player $i$'s posterior distribution on $(x,x_{-i}^\varepsilon)$ conditional on his signal $x_i^\varepsilon$.

The timing of $\Gamma^\varepsilon$ is as follow. First, Nature draws a true $x$. Second, each player $i$ receives his private signal $x_i^{\varepsilon}$ of $x$. Third, all players simultaneously choose their actions. Lastly, payoffs are realized according to the true $x$ and the actions chosen by all players. We note that players play once and then the game is over; see \cite{angeletos2007dynamic} and \cite{chassang2010fear} for anlyses of dynamic global games.

\subsection{Concepts and notation}

\textit{Strategies.} A strategy $p_i$ for player $i$ in $\Gamma^{\varepsilon}$ is a function that assigns to any $x_i^{\varepsilon}\in[\underline{X}-\varepsilon,\overline{X}+\varepsilon]$ a probability $p_i(x_i^{\varepsilon})\geq0$ with which the player chooses action $a_i=1$ when they observe $x_i^{\varepsilon}$. Write $p=(p_1,p_2,...,p_N)$ for a strategy vector for all player, and $p_{-i}=(p_j)_{j\neq i}$ for the vector of strategies for all players but $i$. A strategy vector $p$ is \textit{symmetric} if for every $i,j\in\N$ and every signal $x^\varepsilon$ one has $p_i(x^\varepsilon)=p_j(x^\varepsilon)$. Conditional on the strategy vector $p_{-i}$ and a private signal $x_i^{\varepsilon}$, the expected incentive to play 1 for player $i$ is given by:
\begin{equation*}\label{eq_exp_inc}
	u^{\varepsilon}_i(p_{-i}\mid x_i^{\varepsilon}):=\int u_i(p_{-i}(x^{\varepsilon}_{-i})\mid x)\, \wrt F^{\varepsilon}_i(x,x^{\varepsilon}_{-i}\mid x_i^{\varepsilon}).
\end{equation*}
When no confusion can arise, we refer to the expected incentive $u^{\varepsilon}_i(p_{-i}\mid x_i^{\varepsilon})$ simply as a player's incentive.

\textit{Increasing strategies.} For $X\in\mathbb{R}$, let $p_i^X$ denote the particular strategy such that $p_i^X(x_i^{\varepsilon})=0$ for all $x_i^{\varepsilon}<X$ and $p_i^X(x_i^{\varepsilon})=1$ for all $x_i^{\varepsilon}\geq X$. The strategy $p_i^X$ is called an \textit{increasing strategy with switching point} $X$. Let $p^X=(p_1^X,p_2^X,...,p_N^X)$ denote the strategy vector of increasing strategies with switching point $X$, and $p_{-i}^X=(p_j^X)_{j\neq i}$. Generally, for a vector of real numbers $y=(y_i)$ let $p^y=(p_i^{y_i})$ be a (possibly asymmetric) increasing strategy vector, and $p_{-i}^y=(p_j^{y_j})_{j\neq i}$. 

\textit{Strict dominance.} The action $a_i=1$ is strictly dominant at $x_i^{\varepsilon}$ if $u_i^{\varepsilon}(p_{-i}\mid x_i^{\varepsilon})>0$ for all $p_{-i}$. Similarly, the action $a_i=0$ is strictly dominant (in the global game $G^{\varepsilon}$) at $x_i^{\varepsilon}$ if $u_i^{\varepsilon}(p_{-i}\mid x_i^{\varepsilon})<0$ for all $p_{-i}$. When $a_i=\alpha$ is strictly dominant, the action $a_i=1-\alpha$ is said to be strictly dominated.

\textit{Conditional dominance.} Let $L$ and $R$ be real numbers. The action $a_i=1$ is said to be dominant at $x_i^{\varepsilon}$ conditional on $R$ if $u_i^{\varepsilon}(p_{-i}\mid x_i^{\varepsilon})>0$ for all $p_{-i}$ with $p_j(x_j^{\varepsilon})=1$ for all $x_j^{\varepsilon}>R$, all $j\neq i$. Similarly, the action $a_i=0$ is dominant at $x_i^{\varepsilon}$ conditional on $L$ if $u_i^{\varepsilon}(p_{-i}\mid x_i^{\varepsilon})<0$ for all $p_{-i}$ with $p_j(x_j^{\varepsilon})=1$ for all $x_j^{\varepsilon}>L$, all $j\neq i$. Note that $a_i=1$ is strictly dominant at $x_i^{\varepsilon}$ conditional on $R$ if and only if $u_i^{\varepsilon}(p_{-i}^R\mid x_i^{\varepsilon})>0$. Similarly, if $a_i=0$ is strictly dominant at $x_i^{\varepsilon}$ conditional on $L$ then it must hold that $u_i^{\varepsilon}(p_{-i}^L\mid x_i^{\varepsilon})<0$.

\textit{Iterated elimination of strictly dominated strategies.} The solution concept in this paper is iterated elimination of strictly dominated strategies (IESDS). Eliminate all pure strategies that are strictly dominated, as rational players may be assumed never to pursue such strategies. Next, eliminate a player's pure strategies that are strictly dominated if all other players are known to play only strategies that survived the prior round of elimination; and so on. The set of strategies that survive infinite rounds of elimination are said to survive IESDS.

\section{Optimal Subsidies}\label{sec_subsidies}
\subsection{The Planner's Problem}
Next we introduce a social Planner whose problem is to implement (in a way made precise shortly) coordination on $\overline{a}$ whenever $x>\tilde{x}$, where $\tilde{x}\in\mathcal{X}$ is the \textit{critical state} which she -- the Planner -- chooses.\footnote{We consider the generalized implementation problem in which $\tilde{x}=(\tilde{x}_i)\in\mathcal{X}^N$ is a vector of, possibly distinct, real numbers such that the planner seeks to induce player $i$ to play $1$ whenever $x>\tilde{x}_1$ in Section~\ref{sec_as}.} The Planner faces two constraints. First, she cannot condition her policy on the realization of $x$ or players' signals thereof; this assumption is customary in the literature on policy design in global games \citep[\textit{cf.}][]{sakovics2012matters,leister2022social}.\footnote{Though customary, this assumption is not without loss. As \cite{angeletos2006signaling} demonstrate, if the Planner can decide upon her policy after learning $x$, the endogenous information generated by her intervention can re-introduce equilibrium multiplicity.} One interpretation is that the Planner must commit to her policy before Nature draws a true $x$ and cannot change her policy afterward.

The second constraint upon the Planner's problem has to do with the kinds of policies she can use. We assume that the Planner cannot coordinate players on her preferred equilibrium in a multiple equilibria setting. Instead, she has to rely on simple subsidies (or taxes) to create the appropriate incentives. The focus on simple instruments also means that policies cannot condition directly upon other players' actions. These are standard assumptions in the literature \citep{segal2003coordination,winter2004incentives,bernstein2012contracting,sakovics2012matters,halac2020raising}. 

To streamline the narrative, we henceforth focus on subsidies as the Planner's policy instrument. Let $s_i$ denote the subsidy paid to a player $i$ who chooses $a_i=1$. Conditional on the subsidy $s_i$, player $i$'s incentive to choose 1 becomes
\begin{equation*}
	u_i(a_{-i}\mid x,s_i)=u_i(a_{-i}\mid x)+s_i,
\end{equation*}
and the expected incentive, given the signal $x_i^\varepsilon$ and a strategy vector $p_{-i}$, is
\begin{equation}\label{eq_exp_incentive_sub}
	\begin{split}
		u_i^\varepsilon(p_{-i}\mid x_i^\varepsilon,s_i)&=\int u_i(p_{-i}(x_{-i}^\varepsilon)\mid x,s_i)\, \wrt F_i^{\varepsilon}(x,x^{\varepsilon}_{-i}\mid x_i^{\varepsilon})\\
		&=\int \left[u_i(p_{-i}(x_{-i}^\varepsilon)\mid x)+s_i\right] \wrt F_i^{\varepsilon}(x,x^{\varepsilon}_{-i}\mid x_i^{\varepsilon})\\
		&=u_i^\varepsilon(p_{-i}\mid x_i^\varepsilon)+s_i.
	\end{split}
\end{equation}
It is clear that a tax equal to $s_i$ on playing 0 has the same effect on incentives. Note that \eqref{eq_exp_incentive_sub} assumes observability of $a_i$; this assumption is maintained throughout most of the analysis. Section~\ref{sec_pa} considers principal-agent problems in which the vector of actions $a$ is unobserved.

\subsection{Unique Implementation}
Given a vector of subsidies $s=(s_i)$, let $\Gamma^\varepsilon(s)$ denote the game $\Gamma^\varepsilon$ in which the Planner publicly commits to paying each player $i$ who plays 1 a subsidy $s_i\in s$. Since the Planner cannot condition her policy on $x$, and because players choose their actions before learning the true value of $x$, we must define implementation in terms of players' signals. Henceforth, the vector of subsidies $\tilde{s}$ is said to \textit{implement} coordination on $\overline{a}$ for all $x>\tilde{x}$ if $p^{\tilde{x}}=(p_i^{\tilde{x}})$ is the unique Bayesian Nash equilibrium of $\Gamma^\varepsilon(\tilde{s})$. The focus on unique equilibrium implementation is in keeping with the broader literature on policy design in coordination games \citep[\textit{cf.}][]{segal1999contracting,segal2003coordination,segal2000naked,sakovics2012matters,bernstein2012contracting,halac2020raising,halac2021rank,halac2022addressing}. Note that, as $\varepsilon\to0$, the working definition of implementation also solves the Planner's problem as originally formulated: if $\varepsilon\to0$ then for all $x>\tilde{x}$ each player $i$ receives a signal $x_i^\varepsilon>\tilde{x}$ which in the unique equilibrium $p^{\tilde{x}}$ implies that players coordinate on $\overline{a}$ for all $x>\tilde{x}$.

Take some critical state $\tilde{x}\in\mathcal{X}$. Given $\tilde{x}$, let $s^*(\tilde{x})=(s_i^*(\tilde{x}))$ denote the subsidy scheme such that each $s_i^*(\tilde{x})\in s^*(\tilde{x})$ is given by
\begin{equation}\tag{$*$}\label{eq_opt_sub_asym}
	s^*_i(\tilde{x})=c_i-\tilde{x}-\sum_{n=0}^{N-1}\frac{w_i(n)}{N}.
\end{equation}
Let us write $\mathcal{B}_r(s^*(\tilde{x}))$ for the open ball with radius $r$ centered at $s^*(\tilde{x})$. The main result of the paper is the following theorem.
\begin{theorem}\label{thm_subsidy}
	Let $\tilde{x}\in\mathcal{X}$. If $\varepsilon$ is sufficiently small, then:
	\begin{itemize}
		\item[(i)] There exists a unique subsidy scheme $\tilde{s}=(\tilde{s}_i)$ that implements $p^{\tilde{x}}$;
		\item[(ii)] For all $r>0$, the scheme $\tilde{s}$ is contained in $\mathcal{B}_r(s^*(\tilde{x}))$.
	\end{itemize}
\end{theorem}
The optimal subsidy scheme $\tilde{s}$ admits a number of notable properties, some of which are best understood with the analysis in mind. We therefore defer a discussion of the properties of $\tilde{s}$ to Section~\ref{sec_discussion}.

We observe that Theorem~\ref{thm_subsidy} holds for all continuous densities $f$ and $g$. Thus, the informational requirements imposed upon the Planner are slim. Moreover, the condition that $\varepsilon$ be sufficiently small is necessary to permit an analysis of $\Gamma^\varepsilon(s)$ ``as if'' the common prior $g$ were uniform. The following corollary to Theorem~\ref{thm_subsidy} is immediate from our proof.
\begin{corollary}\label{cor_thm}
	If the common prior $g$ is uniform and the noise distribution $f$ is symmetric, then Theorem~\ref{thm_subsidy} holds true for all $\varepsilon>0$.
\end{corollary}
Uniform common priors are often assumed in the applied literature on global games \citep[\textit{cf.}][]{morris1998unique,angeletos2006signaling,angeletos2007dynamic,sakovics2012matters}. In Appendix~\ref{app_small} we show why $\Gamma^\varepsilon(s)$ behaves ``as if'' $g$ were uniform when $\varepsilon$ is small.

The analysis will reveal that Theorem~\ref{thm_subsidy} remains valid under a slightly more general definition of implementation. We show that $\tilde{s}$ is the unique subsidy scheme such that $p^{\tilde{x}}$ is the unique strategy vector that survives IESDS in $\Gamma^\varepsilon(\tilde{s})$. Implementation as a unique strategy vector that survives IESDS is more general than implementation as a unique Bayesian Nash equilibrium because the former implies the latter but the reverse implication is not necessarily true. In this sense, as in \cite{sandholm2002evolutionary,sandholm2005negative}, we need not impose that players play an equilibrium of the game but could depart from more primitive assumptions on players' strategic sophistication by requiring that none play a strategy that is iteratively dominated. Equilibrium play would then be obtained as a result, rather than an assumption, of the analysis.

Lastly, we note that Theorem~\ref{thm_subsidy} is a positive result: given the Planner's choice of $\tilde{x}$, Theorem~\ref{thm_subsidy} characterizes the unique subsidy scheme that implements $p^{\tilde{x}}$. We are agnostic about her exact motivation for choosing $\tilde{x}$. Several intuitive objectives could underly her choice. For example, the Planner might face a budget constraint $B>0$ and seek to maximize the prior probability of coordination on $(1,1,...,1)$ given her budget; this is the problem of the ``adoption-maximization planner'' in \cite{leister2022social}.\footnote{\cite{sakovics2012matters} consider the related problem of a planner who seek to maximize the probability of coordination on $(1,1,...,1)$ at minimal cost; however, the planner in their model is not bound by an explicit budget constraint.} Alternatively, the Planner might seek to implement the equilibrium $p$ that maximizes expected welfare, where welfare is some increasing function of players' payoffs; this is the problem of the ``welfare-maximization planner'' in \cite{leister2022social}. As said, this paper remains agnostic as to the Planner's decision-making process -- all we do is show how, conditional on her choice of $\tilde{x}$, she can implement $p^{\tilde{x}}$.

\section{Analysis}\label{sec_main}
The plan for this section is as follows. We first show that for any vector of subsidies $s$ there exists a unique vector of real numbers $x(s)=(x_i(s))$ such that the increasing strategy vector $p^{x(s)}$ is the unique strategy vector that survives IESDS in $\Gamma^\varepsilon(s)$. Then we demonstrate that the strategy vector $p^{x(s)}$ is also the unique Bayesian Nash equilibrium of $\Gamma^\varepsilon(s)$. We use this, and some minor technical results, to derive the unique subsidy scheme $\tilde{s}$ that implements $p^{\tilde{x}}$.

\subsection{Monotonicities}

Suppose that all of player $i$'s opponents are known to play increasing strategies, say $p_{-i}^{y}=(p_j^y)_{j\neq i}$. Then his incentive $u_i^\varepsilon$ to play 1 satisfies a two intuitive monotonicity properties.

\begin{lemma}\label{lemma_increasing}
	Given is a vector of real numbers $y=(y_i)$ and the associated increasing strategy vector $p^y=(p_i^{y_i})$. Then,
	\begin{itemize}
		\item[(i)] $u^{\varepsilon}_i(p_{-i}^y\mid x_i^{\varepsilon})$ is monotone increasing in $x_i^\varepsilon$;
		\item[(ii)] $u^{\varepsilon}_i(p_{-i}^y\mid x_i^{\varepsilon})$ is monotone decreasing in $y_j$, all $j\in\N\setminus\{i\}$.
	\end{itemize}
\end{lemma}

Part (i) of Lemma~\ref{lemma_increasing} says that a player's incentive to play 1 is increasing in his type $x_i^\varepsilon$ when his opponents play increasing strategies. There are two sides to this. First, taking as given the vector of actions $a_{-i}$, a player's expected payoff to playing 1 is linearly increasing in $x_i^\varepsilon$; hence, his expected incentive is increasing in his signal $x_i^\varepsilon$. Second, as $x_i^\varepsilon$ increases player $i$'s posterior distribution on the hidden state $x$ and, therefore, the signals of his opponents shifts to the right. If his opponents play increasing strategies, this also shifts his distribution of the aggregate action to the right which, because externalities are increasing in the aggregate action, further raises his incentive to play 1. Note that monotonicity of $u^{\varepsilon}_i(p_{-i}^y\mid x_i^{\varepsilon})$ in $x_i^\varepsilon$ depends upon $p_{-i}^y$ being increasing; for generic $p_{-i}$, $u^{\varepsilon}_i(p_{-i}\mid x_i^{\varepsilon})$ can be locally decreasing in $x_i^\varepsilon$.

Part (ii) of Lemma~\ref{lemma_increasing} says that the incentive to play 1 of a player $i$ whose opponents play increasing strategies is decreasing in the switching point of each of these increasing strategies. For given signal $x_i^\varepsilon$, the probability player $i$ attaches to the event that his opponent $j$ receives a signal $x_j^\varepsilon>y_j$ and thus, in $p_j^{y_j}$, plays 1 is decreasing in $y_j$. Therefore player $i$'s incentive to play 1 is decreasing in the switching $y_j$.

The analysis relies repeatedly upon Lemma~\ref{lemma_increasing} for much of the heavy lifting. While a focus on increasing strategies seems natural in $\Gamma^\varepsilon(s)$ \citep[\textit{cf}.][]{angeletos2007dynamic}, the results in Lemma~\ref{lemma_increasing} are of true pratical use only once the focus on increasing strategies has been properly defended. The next sextion provides such a justification; Lemma~\ref{lemma_unique_limit} pushes it to its ultimalte conclusion.

\subsection{Subsidies, Strategies, Selection}\label{sec_formal}
Recall that $x_i^N$ and $x_i^0$ demarcate strict dominance regions for player $i$: when $x<x_i^N$ [$x>x_i^0$], playing $0$ [playing $1$] is strictly dominant for player $i$ in $\Gamma(x)$. A subsidy $s_i$ to player $i$ shifts these boundaries to $x_i^N-s_i$ and $x_i^0-s_i$, respectively. In the game of incomplete information $\Gamma^\varepsilon$, the boundaries for strict dominance in terms of a player's signals instead are $x_i^N-s_i-\varepsilon/2$ and $x_i^0-s_i+\varepsilon/2$, respectively. That is, for all $x_i^\varepsilon>x_i^0-s_i+\varepsilon/2$ player $i$ knows that any true state $x$ consistent with his signal satisfies $x>x_i^0-s_i$, in which case playing 1 is strictly dominant. To make the following arguments work, we must assume that $\overline{X}\geq \overline{x}-s_i+\varepsilon/2$ and $\underline{X}\leq\underline{x}-s_i-\varepsilon/2$ for all $i\in\N$, imposing a joint restriction on permissible values of $(\underline{X},\overline{X},s)$ given $\varepsilon$. This assumption is henceforth maintained.


Per the foregoing argument, given the assumption that $\overline{X}\geq \overline{x}-s_i+\varepsilon/2$, we know that $u_i^\varepsilon(p_{-i}\mid \overline{X},s_i)>0$ for all $p_{-i}$. In particular, therefore, one has
\begin{equation*}
	u_i^\varepsilon(p_{-i}^{\overline{X}}\mid \overline{X},s_i)>0.
\end{equation*}
Let $r^1_i$ be the solution to
\begin{equation*}\label{eq_ID_first}
u_i^\varepsilon(p_{-i}^{\overline{X}}\mid r^1_i,s_i)=0.
\end{equation*}
To any player $i$, the action $a_i=1$ is strictly dominant at all $x_i^\varepsilon>r^1_i$ conditional on $\overline{X}$; denote $r^1:=(r^1_i)$. It is clear that $r_i^1$ depends upon the subsidy $s_i$, but for brevity we leave this dependence out of the notation for now. From Lemma~\ref{lemma_increasing} follows that $r_i^1<\overline{X}$ for all $i$.

Player $i$ knows that no player $j$ will pursue a strategy $p_j< p_j^{r_j^1}$ since such a strategy is iteratively strictly dominated. Now define $r^2=(r_i^2)$ as the signal that solves
\begin{equation*}\label{eq_step2_id}
	u_i^\varepsilon(p_{-i}^{r^1}\mid r^2_i,s_i)=0,
\end{equation*}
for all $i$. Because $p_i^{\overline{X}}$ is strictly dominated for every $i$, the any strategy $p_i<p_i^{r_i^1}$ is iteratively strictly dominated for all $i$, which in turn implies that any $p_i<p_i^{r_i^2}$ is iteratively dominated. This argument can -- and should -- be repeated indefinitely. We obtain a sequence $\overline{X}=r^0_i,r^1_i,...$, all $i$. For any $k$ and $r^k_i$ such that $u_i^\varepsilon(p^{r^k_{-i}}\mid r^k_i,s_i)>0$, there exists $r^{k+1}_i$ that solves $u_i^\varepsilon(p_{-i}^{r^k}\mid r^{k+1}_i,s_i)=0$. Induction on $k$, using Lemma~\ref{lemma_increasing}, reveals that $r^{k+1}_i<r^k_i$ for all $k\geq0$. Moreover, we know that $r^k_i\geq\underline{X}$ for all $k$. It follows that the sequence $(r^k_i)$ is monotone and bounded. Such a sequence must converge; let $r_i(s)$ denote its limit and define $r(s):=(r_i(s))$. By construction, $r(s)$ solves
\begin{equation*}\label{eq_R*}
	u_i^\varepsilon\left(p_{-i}^{r(s)}\mid r_i(s),s_i\right)=0.
\end{equation*}

A symmetric procedure should be carried out starting from low signals, eliminating ranges of $x_i^\varepsilon$ for which playing 1 is strictly (iteratively) dominated. For every player $i$ this yields an increasing and bounded sequence $(l^k_i)$ whose limit is $l_i(s)$, and $l(s):=(l_i(s))$. The limit $l(s)$ solves $u_i^\varepsilon(p_{-i}^{l(s)}\mid l_i(s),s_i)=0$ for all $i$.

It is clear from the foregoing construction that a strategy $p_i$ survives IESDS if and only if $p_i^{r_i(s)}(x_i^\varepsilon)\leq p_i(x_i^\varepsilon)\leq p_i^{l_i(s)}(x_i^\varepsilon)$ for all $x_i^\varepsilon$. We are particularly interested in games in which the points $l_i(s)$ and $r_i(s)$ converge to a common limit $x(s):=(x_i(s))$ that, hence, is the (essentially) unique solution to
\begin{equation}\label{eq_x(s)}
	u_i^\varepsilon\left(p_{-i}^{x(s)}\mid x_i(s),s_i\right)=0
\end{equation}
for all $i\in\N$. To work in such an environment, we must assume $\varepsilon$ to be sufficiently small.
\begin{lemma}\label{lemma_unique_limit}
For all $\delta>0$, there exists $\varepsilon(\delta)>0$ such that $r_i(s)-l_i(s)<\delta$ for all $\varepsilon\leq\varepsilon(\delta)$ and all $i\in\N$.
\end{lemma}
We note that assuming $\varepsilon\to0$ is sufficient but not, in general, necessary to obtain convergence to the common limit $x(s)$; for example, when $g$ is uniform we have $l_i(s)=r_i(s)$ for all $\varepsilon>0$.


Given a subsidy scheme $s$ and small enough $\varepsilon$, there is a unique increasing strategy vector $p^{x(s)}$ that survives IESDS in $\Gamma^\varepsilon(s)$. We next establish that the relation between $x(s)$ and $s$ is one-to-one: given any $\hat{x}$, there is a unique subsidy scheme $\hat{s}$ such that $p^{\hat{x}}$ is the unique strategy vector that survives IESDS in $\Gamma^\varepsilon(\hat{s})$.

\begin{lemma}\label{lemma_unique_subsidy}
	Given is a vector of real numbers $\hat{x}=(\hat{x}_i)$ and $\varepsilon$ sufficiently small. There is a unique subsidy scheme $\hat{s}=(\hat{s}_i)$ such that $x(\hat{s})=\hat{x}$.
\end{lemma}

\subsection{Implementation and Characterization}
Recall that a strategy vector $p=(p_1,p_2,...,p_N)$ is a Bayesian Nash Equilibrium (BNE) of $\Gamma^{\varepsilon}(s)$ if for any $p_i$ and $x_i^{\varepsilon}$ it holds that:
\begin{equation}\label{eq_def_BNE}
	p_i(x_i^{\varepsilon})\in\argmax_{a_i\in\{0,1\}}\pi_i^{\varepsilon}(a_i,p_{-i}\mid x_i^{\varepsilon},s_i),
\end{equation}
where $\pi_i^{\varepsilon}(a_i,p_{-i}\mid x_i^{\varepsilon}):=\int \pi_i(a_i,p_{-i}(x^{\varepsilon}_{-i})\mid x)\, \wrt F_i^{\varepsilon}(x,x^{\varepsilon}_{-i}\mid x_i^{\varepsilon})$. It follows immediately that $p^{x(s)}$ is a BNE of $\Gamma^\varepsilon(s)$. Lemma~\ref{thm_BNE} strengthens this result and establishes that $p^{x(s)}$ is the \textit{only} BNE of $\Gamma^\varepsilon(s)$.

\begin{lemma}\label{thm_BNE}
	Given is $s$ and $\varepsilon$ sufficiently small. The essentially unique Bayesian Nash equilibrium of $\Gamma^\varepsilon(s)$ is $p^{x(s)}$. In particular, if $p$ a BNE of $\Gamma^\varepsilon(s)$ then any $p_i\in p$ satisfies $p_i(x_i^{\varepsilon})=p_i^{x_i(s)}(x_i^{\varepsilon})$ for all $x_i^{\varepsilon}\neq x_i(s)$ and all $i$.  
\end{lemma}

We know that for any subsidy scheme $s$ and small enough $\varepsilon$ the increasing strategy vector $p^{x(s)}$ is the unique BNE of $\Gamma^\varepsilon(s)$. From Lemma~\ref{lemma_unique_limit}, we furthermore know that there is a unique subsidy scheme $\tilde{s}$ such that $x_i(\tilde{s})=\tilde{x}$ for all $i$. It follows that the subsidy scheme $\tilde{s}$ that implements $p^{\tilde{x}}$ exists and is unique, provided we set $\varepsilon$ sufficiently small. This proves part (i) of Theorem~\ref{thm_subsidy}. All that is left to do now is to characterize $\tilde{s}$. We rely on the following lemma.

\begin{lemma}\label{lemma_monotone}
	For all $\delta>0$ there exists $\varepsilon(\delta)>0$ such that
	\begin{equation}\label{eq_gain_p^X_X}
		\left| u_i^{\varepsilon}\left(p^X_{-i}\mid X,s_i\right)-\left[X+\sum_{n=0}^{N-1}\frac{w_i(n)}{N}-c_i+s_i\right]\right|<\delta
	\end{equation}
	for $\varepsilon\leq\varepsilon(\delta)$ and all $X$ such that $\underline{X}+\varepsilon\leq X\leq \overline{X}-\varepsilon$.
\end{lemma}
If his opponents all play the same increasing strategy $p_j^X$, then upon observing the threshold signal $x_i^\varepsilon=X$ player $i$'s belief over the aggregate action $\sum_{j\neq i}a_j$ is uniform. Convergence to uniform strategic beliefs is a common property in global games; see Lemma 1 in \cite{sakovics2012matters} for a reference in the context of policy design.

Recall that, if $x(s)$ is the vector of switching points such that $p^{x(s)}$ is the unique BNE of $\Gamma^\varepsilon(s)$, then $x_i(s)$ solves \eqref{eq_x(s)} for all $i$. Imposing now that $\tilde{s}$ be such that $x_i(\tilde{s})=\tilde{x}$ for all $i\in\N$, one obtains
\begin{equation}\label{eq_id_sub}
	u_i^{\varepsilon}\left(p^{\tilde{x}}_{-i}\mid \tilde{x},\tilde{s}_i\right)=0
\end{equation}
as the $N$ identifying conditions for the subsidy scheme $\tilde{s}=(\tilde{s}_i)$ that implements $p^{\tilde{x}}$. Using the result in Lemma~\ref{lemma_monotone} when $X=\tilde{x}$ and solving \eqref{eq_id_sub} for $\tilde{s}_i$ establishes that for all $r>0$ there exists $\varepsilon(r)>0$ such that
\[\left|\tilde{s}_i-s_i^*(\tilde{x})\right|<r\]
for all $\varepsilon\leq\varepsilon(r)$ and all $i\in\N$. Hence, the subsidy scheme $\tilde{s}$ is contained in $\mathcal{B}_r(s^*(\tilde{x}))$ for any radius $r>0$ provided we choose $\varepsilon$ sufficiently small. This proves part (ii) of Theorem~\ref{thm_subsidy}.

\subsection{Discussion}\label{sec_discussion}
Our results characterize the subsidy scheme $\tilde{s}$ a Planner must commit to when seeking to implement $p^{\tilde{x}}$ among rational players. Let us discuss several properties of this policy.

First, optimal subsidies are modest relative to the Planner's goal: $\tilde{s}_i$ does not make $p_i^{\tilde{x}}$ strictly dominant for any player $i$.
The sufficiency of modest subsidies is the consequence of a strategic \textit{unraveling effect} of subsidies in coordination games. A subsidy to player $i$ raises his incentive to play 1. In a coordination game, the increased incentive of player $i$ raises the incentive of player $j$ to play 1. The increase in $j$'s incentive in turn makes playing 1 even more attractive to player $i$, and so on. Under common knowledge of the subsidy, what obtains is a indefinitely compounded positive feedback look, the unraveling effect; see Figure~\ref{fig_indirect} in the Introduction. Because of the unraveling effect, even seemingly minor subsides can go a long way toward solving the Planner's problem. This feature of $\tilde{s}$ is a key counterpoint to several well-known results in the literature on policy design in coordination problems that stress optimality of subsidizing at least some players to strict dominance \citep{segal2003coordination,winter2004incentives,bernstein2012contracting,sakovics2012matters}.

%

Second, symmetric players are offered identical subsidies. This symmetry deviates from a number of other notable proposals including a divide-and-conquer policy \citep[\textit{cf.}][]{segal2003coordination,bernstein2012contracting} and the incentive schemes studied in \cite{winter2004incentives} and \cite{halac2020raising}.\footnote{\cite{onuchic2023signaling} also show that ``identical agents'' may be compensated asymmetrically in equilibrium; however, though identical in the payoff-relevant sense their players may still vary in payoff-irrelevant ``identifies''.} The policies derived in \cite{sakovics2012matters} and \cite{halac2021rank} also treat symmetric players identically.

Third, subsidies target all players and are globally continuous in model parameters. The characterization in \eqref{eq_opt_sub_asym} establishes global continuity of $\tilde{s}_i$ in all the parameters upon which it depends. While conditional on policy treatment the optimal subsidies in \cite{sakovics2012matters} are continuous in the relevant model parameters as well, changes in one player's parameters could affect whether or not said player is targeted, causing a discrete jump in subsidies received. Similarly, subsidies are continuous conditional on a player's position in the policy ranking in a divide and conquer mechanism \citep{segal2003coordination,bernstein2012contracting}; however, a player's position in the optimal ranking is affected by a change in its parameters, which can lead to discrete jumps in subsidy entitlement.

Fourth, the subsidy scheme $\tilde{s}$ is unique. In the complete information environments considered by \cite{segal2003coordination}, \cite{winter2004incentives}, and \cite{bernstein2012contracting} the optimal policy is not unique when (some) players are symmetric. In the incomplete information environments considered by \cite{sakovics2012matters} and \cite{halac2021rank}, the optimal policy is unique. Note, however, that the results in \cite{sakovics2012matters} and \cite{halac2021rank} establish uniqueness of the policy that minimizes the expected cost of implementing a given equilibrium; in their models, there still exist other, more expensive policies that implement the same equilibrium. In contrast, Theorem~\ref{thm_subsidy} establishes that only one policy can implemenet a given equilibrium of the game studied here.

Fifth, subsidies are increasing in $c_i$, the opportunity cost of playing 1. Given $x$, the cost of playing 1 is inreasing in $c_i$; hence, to induce coordination on 1 subsidies should increase as the cost $c_i$ rises. This property is intuitive and shared (conditional on policy treatment and/or ranking) by many recent contributions on policy design in coordination problems \citep{segal2003coordination,winter2004incentives,sakovics2012matters,bernstein2012contracting,halac2020raising,halac2021rank}.

Sixth, subsidies are decreasing in $\tilde{x}$, the threshold for coordination on 1 targeted by the Planner. All else equal, a player's incentive to play 1 is increasing in his signal $x_i^\varepsilon$. Hence, for higher signals a player needs less subsidy to induce him to play 1. One can interpret $\tilde{x}$ as an inverse measure of the Planner's ambition: the higher is $\tilde{x}$, the lower is the prior probability that coordination on 1 will be achieved. In this interpretation, being ambitious is costly: assuming coordination on 1 is indeed achieved, total spending on subsidies is increasing in the Planner's ambition (decreasing in $\tilde{x}$). The same is true in \cite{sakovics2012matters}.

Seventh, subsidies are decreasing in spillovers, i.e.\ $\partial \tilde{s}_i/\partial w_i(n)<0$. When observing the threshold signal $\tilde{x}$, a player $i$'s belief over the aggregate action $A_{-i}$ is uniform; in particular, therefore, he assigns strictly positive probability to the event that $A_{-i}=n$ for all $n=0,1,...,N-1$. If $w_i(n)$ increases, the \textit{expected} spillover a player expects to enjoy upon playing 1 is hence greater. This raises his incenive to play 1 and, for given $\tilde{x}$, the subsidy required to make him willing to do so is smaller. Given a ranking of players, subsides for each player (except the first-ranked) are also decreasing in spillovers in a divide-and-conquer policy \citep{segal2003coordination,bernstein2012contracting}. The optimal subsidies in \cite{sakovics2012matters} are not generally decreasing in spillovers, except insofar as players who benefit less from project success are more likely to be targeted.

Eigth, players do not necessarily have symmetric payoffs in the equilibrium. Thus, while their equilibrium strategies are symmetric by construction, an optimal subsidy does not guarantee symmetry in equilibrium.

\section{Special Cases and Extensions}\label{sec_special}
Throughout this section, we assume that $\varepsilon\to0$ to simplify the statements of results.

\subsection{Principal-Agent Problems}\label{sec_pa}

There is an organizational project that involves $N$ tasks each performed by one agent $i\in\N$. Each agent $i$ decides whether to work ($a_i=1$) towards completing his task or shirk ($a_i=0$). The cost of working to agent $i$ is given by $c_i>0$. Success of the project depends upon the decisions of all agents through a production technology $q:\{0,1,...,N\}\to[0,1]$, where $q(n)$ is the probability of success given that $n$ agents work. As in \cite{winter2004incentives} and \cite{halac2021rank}, we assume that $q(n+1)>q(n)$ for all $n\leq N-1$.

A principal offers contracts that specify rewards $v=(v_i)$ to agents contingent on project success; if the project fails, all agents receive zero. We assume that agents' work effort is their private knowledge -- any rewards the principal offers can condition only upon project success. An agent who shirks gets payoff $-x$. We interpret $x$ generally as an uncertain fundamental that determines agents' payoffs, see also \cite{halac2022addressing} for a model of contracting under fundamental uncertainty.
\begin{proposition}\label{cor_pa}
	Consider a principal-agent problem in which the principal offers rewards $(\tilde{v}_i)$ to implement $p^{\tilde{x}}$ as the unique equilibrium of the game. For each $i\in\N$, the reward $\tilde{v}_i$ is given by
	\begin{equation}\label{eq_PA}
		\bar{q}\cdot \tilde{v}_i=c_i-\tilde{x},
	\end{equation}
where $\bar{q}:=\sum_{n=0}^{N-1}\frac{q(n+1)-q(n)}{N}$.
\end{proposition}

Note that for $x=0$, the payoffs in this model exactly replicate those of the canonical problem studied by \cite{winter2004incentives}. Interestingly, the analysis under uncertainty fails to yield Winter's prescription that optimal contracts are inherently discriminatory and should reward identical agents asymmetrically. In this paper, symmetric agents receive identical rewards.

It is interesting to compare Proposition~\ref{cor_pa} to a result in Halac et al.\ (\citeyear{halac2021rank}, Theorem 2 and Corollary 1 in particular). These authors consider the problem of a Planner who offers agents rewards in a \textit{ranking scheme}. In a ranking scheme, agents first are ranked; conditional on his ranking, agent $i$ is then offered a reward that makes him indifferent between working and shirking provided all agents who are ranked below [above] him work [shirk]. Moreover, contract offers a private so that agents face uncertainty about their ranking. For the case of symmetric agents, \cite{halac2021rank} establish that an optimal ranking scheme induces uniform beliefs about each agent's ranking. One can interpret Proposition~\ref{cor_pa} along similar lines: if an agent is ranked $n$-th and believes that all agents ranked below [above] him work [shirk], the reward that is necessary to make him work for all $x_i^\varepsilon>\tilde{x}$ is $v_i(n)=(c_i-\tilde{x})/(q(n+1)-q(n))$. Thus, if an agent has uniform beliefs about his own ranking the necessary reward becomes $\sum_{n=0}^{N-1}v_i(n)/N$, which is exactly the optimal reward $\tilde{v}_i$ given in Proposition~\ref{cor_pa}. Note that, in our analysis, the uniform belief over $n$ is also valid when agents are asymmetric.

\subsection{Heterogeneous Externalities}\label{sec_heterogeneous}

The main analysis assumes that only the aggregate action $A_{-i}$ matters for the externality other players impose upon player $i$. We relax this assumption here. In particular, we allow that the externality $w_i(a_{-i})$ depends upon the specific vector $a_{-i}$ played. We maintain a focus on games with strategic complementarities and assume that if $a_{-i}''\geq a_{-i}$, then $w_i(a_{-i}'')\geq w_i(a_{-i})$. Observe that this externality structure encompasses the games in \cite{bernstein2012contracting} and \cite{halac2021rank}, where externalities are allowed to depend upon the subset $M\subseteq\N$ of players who play 1. It also nests the approach in \cite{sakovics2012matters} where externalities depend upon the \textit{weighed} aggregate action. Finally, heterogeneous externalities may arise in coordination games on (directed) graphs \citep{leister2022social}.

Let $a_{-i}^n$ denote an action vector $a_{-i}$ in which exactly $n$ players play 1 (and the remaining $N-n-1$ players play 0). We write $A_{-i}^n$ for the set of all (unique) action vectors $a_{-i}^n$ (i.e.\ $A_{-i}^n:=\{a_{-i}\mid \sum_{a_j\in a_{-i}}a_j=n\}$). Note that there are exactly $\binom{N-1}{n}$ vectors $a_{-i}^n$ in $A_{-i}^n$. For all $i$, define
\[w_i^n:=\frac{\sum_{a_{-i}^N\in A_{-i}^N} w_i\left(a_{-i}^n\right)}{\binom{N-1}{n}}.\]
In words, $w_i^n$ is the expected externality imposed upon player $i$ who expects that $n$ opponents play 1 and believes that any such outcome is equally likely. 

\begin{proposition}\label{cor_het}
	Let $\tilde{x}\in\mathcal{X}$. There exists a unique subsidy scheme $\tilde{s}=(\tilde{s}_i)$ that implements $p^{\tilde{x}}$ in the game  $\Gamma^\varepsilon(s)$ with heterogeneous externalities. The subsidy $\tilde{s}_i$ pursuant to the scheme is given by
	\begin{equation}\label{eq_opt_sub_het}
		\tilde{s}_i=c_i-\tilde{x}-\sum_{n=0}^{N-1}\frac{w_i^n}{N}
	\end{equation}
	for all $i\in\N$ and $\varepsilon$ sufficiently small.
\end{proposition}
We observe that Proposition~\ref{cor_het} doubles down on the uniform strategic beliefs of the game with homogeneous externalities. In a game with heterogeneous externalities, players have uniform beliefs about the total number of opponents $n\in\{0,1,...,N-1\}$ that play 1. Moreover, conditional on the number $n$ of opponents that play 1 a threshold type player also has uniform beliefs about the exact action vector $a_{-i}^n$ played.

In the game with heterogeneous externalities, too, symmetric players receive identical subsidies. This conclusion remains valid if the ucnertainty is negligible. Thus, even a little bit of uncertainty can change the canonical results by \cite{segal2003coordination}, \cite{winter2004incentives}, and \cite{bernstein2012contracting} that optimal contracts are fundamentally discriminatory in coordination games.

\subsection{Asymmetric Targets}\label{sec_as}
It was so far maintained that the Planner seeks to implement a symmetric equilibrium. In many practical situations policies may instead target asymmetric outcomes. We consider such instances here.

Given are two real numbers $\tilde{x}_1$ and $\tilde{x}_2$. Without loss, let $\tilde{x}_1<\tilde{x}_2$. The Planner partitions the player set $\N$ into two subsets $\N_1$ and $\N_2$ such that $\N_1\cup\N_2=\N$. There are $N_1$ players in $\N_1$ and $N_2=N-N_1$ player in $\N_2$. Suppose the Planner seeks to implement the asymmetric equilibrium $\tilde{p}=(p_1^{\tilde{x}_1},p_2^{\tilde{x}_2})$ according to which (with a slight abuse of notation) each player $i\in\N_1$ plays the increasing strategy $p_i^{\tilde{x}_1}$ while each $j\in\N_2$ plays $p_j^{\tilde{x}_2}$. We are agnostic as to the motivations behind such an asymmetric policy goal. We call $\tilde{x}_1$ the low critical state and $\tilde{x}_2$ the high critical state; similarly, we label players in $\N_1$ and $\N_2$ as low state and high state players, respectively.

Let $s^{**}(\tilde{x}_1,\tilde{x}_2)=(s_i^{**}(\tilde{x}_1,\tilde{x}_2))$ be the subsidy scheme such that each $s_i^{**}(\tilde{x}_1,\tilde{x}_2)\in s^{**}(\tilde{x}_1,\tilde{x}_2)$ is given by
\begin{equation}\label{eq_subs_asym}
	s_i^{**}(\tilde{x}_1,\tilde{x}_2)=
	\begin{cases}
		c_{i}-\tilde{x}_1-\sum_{n=0}^{N_1-1}\frac{w_{i}(n)}{N_1}\quad&\text{if }i\in\N_1,\\
		c_{i}-\tilde{x}_2-\sum_{n=N_1}^{N-1}\frac{w_{i}(n)}{N_2}\quad&\text{if }{i}\in\N_2.
	\end{cases}
\end{equation}

\begin{proposition}\label{prop_asym}
	Let $\tilde{x}_1,\tilde{x}_2\in\mathcal{X}$. If $\varepsilon$ is sufficiently small, then:
	\begin{itemize}
		\item[(i)] There exists a unique subsidy scheme $\tilde{s}=(\tilde{s}_i)$ that implements $\tilde{p}$.
		\item[(ii)] For all$r>0$, the scheme $\tilde{s}$ is contained in $\mathcal{B}_r(s^{**}(\tilde{x}_1,\tilde{x}_2))$.
	\end{itemize}
\end{proposition}
Proposition~\ref{prop_asym} is a special case of a more general result in which the Planner partitions the player set into $K$ different subsets $\N_k$ each with their own critical state $\tilde{x}_k$. We prove this more general case in the Appendix.

Part (i) of Proposition~\ref{prop_asym}, existence and uniquenes of $\tilde{s}$ that implements $\tilde{p}$, follows from Lemmas~\ref{lemma_unique_limit}--\ref{thm_BNE}. Part (ii) of the proposition, the characterization of the optimal subsidy scheme $\tilde{s}$, underlines the importance of strategic uncertainty for policy design in coordination games.


Indeed, if one compares to optimal subsidy $\tilde{s}_i$ for some low state player $i\in\N_1$ needed to implement the asymmetric equilibrium $\tilde{p}$, one observes that this subsidy is higher than the subsidy that would be necessary to implement the symmetric equilibrium $p^{\tilde{x}_1}$. This makes sense: if a lows state player observes the low critical state, he should -- by construction -- be indifferent between playing 0 and 1. But upon observing the low critical state $\tilde{x}_1$, player $i$ also knows that no player can have received a signal equal to or above the high critical state $\tilde{x}_2$. In other words, player $i$ knows that none of the high state critical players will play 1. This means that the spillovers player $i$ expects to enjoy are less, which in order for him to be indifferent between playing 0 and 1 implies that a higher subsidy is required compared to a situation in which the Planner seeks to implement the symmetric equilibrium $p^{\tilde{x}_1}$. A similar but opposite logic applies to the subsidies targeting high state players $j\in\N_2$. Upon observing the high critial state, a high state player $j$ knows that all players must have received a signal in excess of the low critical state $\tilde{x}_1$. Hence, player $j$ takes as given that all low state players will play 1. This makes playing 1 more attractice to him and consequently implies that a lower subsidy can make him indifferent between playing 0 and 1 compared to the case in which the Planner would target the symmetric equilibrium $p^{\tilde{x}_2}$.


\subsection{Games of Regime Change}\label{sec_joint_invest}
There is a project in which $N$ investors can invest. The cost of investment to investor $i$ is $c_i>0$. If the project succeeds, an investing investor $i$ realizes benefit $b_i>c_i$. The project is successful if and only if a critical mass of investors invests; specifically, there exists $I\in(0,N)$ such that the project succeeds if and only if $A\geq I$. The payoff to not investing, the outside option, is given by $-x$. Uncertainty about $x$, or more generally about $x-c_i$, can be thought of as any kind of (fundamental) uncertainty that pertains to the cost or benefit of investment \citep{abel1983optimal,pindyck1993investments}. Hoping to attract investment, a Planner offers each invsting investor $i$ a subsidy $s_i$. We are particularly interested in the subsidy scheme $s^0=(s_i^0)$ that implements $p^0$ as regime change games usually normalize the payoff to the outside option to 0 (\textit{cf}.\ \cite{morris1998unique}, \cite{angeletos2007dynamic}, \cite{goldstein2005demand}, \cite{sakovics2012matters}, \cite{edmond2013information},  \cite{basak2020diffusing}, \cite{halac2020raising}). 

\begin{proposition}\label{cor_joint}
	Consider a joint investment problem in which $I\in(0,N)$ is the critical threshold for project success. Let $N-n^*$ be the smallest integer greater than $I$. The subsidy scheme $s^0=(s^0_i)$ that implements $p^0$ is given by
	\begin{equation}\label{eq_threshold_N}
		s_i^0=c_i-\frac{n^*}{N}\cdot b_i
	\end{equation}
	for every $i\in\N$.
\end{proposition}
In the subsidy scheme $s^0$, all investors are subsidized and subsidies are a fraction of their investment costs. The latter is explained through the unraveling effect of policies: if investor $i$ receives an investment subsidy, he is more likely to invest. Anticipating the increased likelihood that $i$ invests, project success becomes more likely and this attracts investment by investor $j$. The greater likelihood that $j$ invests in turn makes investment even more interesting for $i$, and so on. This feedback effect is strong: in (non-trivial) two-player joint investment problems, subsidies are less than half players' investment costs.

While the problem here bears close resemblance to the global game in \cite{sakovics2012matters}, the models differ in fundamental ways that make a direct comparison complicated. In contrast to our game, \cite{sakovics2012matters} do not model prior uncertainty about the efficient outcome of the game; coordinated investment is always the efficient equilibrium of their game. Instead, uncertainty pertains to the critical threshold of investments required to achieve project success, which we assume to be common knowledge. Similarly, conditional on the regime in place, there is certainty about payoffs in \cite{sakovics2012matters}; we instead work with uncertain payoffs even conditional on the regime.\footnote{This distinction applies more generally to the literature on global games of regime change, see \cite{morris1998unique}, \cite{angeletos2007dynamic}, \cite{goldstein2005demand}, \cite{basak2020diffusing}, and \cite{edmond2013information}. Similarly, \cite{kets2022value} (in section 3.3.1, and their Theorem 3.4) also assume that joint investment is the efficient outcome of their game.} It is interesting that these differences, albeit fairly subtle, lead to vastly different policy implications.

An important and, in our view, realistic possibility in the investment problem studied here is that joint investment need \textit{not} be ex post efficient: if $x$ is very low, it can be efficient for all players to not invest and take the outside option. We elaborate upon this issue in the next section.

\subsection{Induced Coordination Failure}\label{sec_symmetric}
The Planner must commit to her policy \textit{before} Nature draws the true state $x$ according to the density $g$. Given our assumptions on $\mathcal{X}$, the support of $g$, this implies that the planner does not know which action vector will be the efficient outcome of the game when she offers her subsidies. She may, hence, commit to subsidizing an action that is ex post inefficient. This possibility warrants policy moderation, as is most simply illustrated in a symmetric game.

Consider the game $\Gamma^\varepsilon$ played among symmetric players such that $c_i=c$ and $w_i(n)=w(n)$ for all $i\in\N$. In symmetric games, $\overline{a}$ is the efficient Nash equilibrium of the complete information game $\Gamma(x)$ for all $x>\underline{x}$ where $\underline{x}=c-w(N-1)$. Let $s^*$ denote the \textit{optimal subsidy} in a symmetric game in the sense that $s^*$ induces coordination on $p^{\underline{x}}$ as the unique Bayesian Nash equilibrium of $\Gamma^\varepsilon(s^*)$. Per the characterization in Theorem~\ref{thm_subsidy}, we have
\begin{equation*}\label{eq_char_opt_sub}
	s^*=w(N-1)-\sum_{n=0}^{N-1}\frac{w(n)}{N}.
\end{equation*}
The following corollary says that $s^*$ is not only sufficient to induce coordination on an efficient equilibrium; subsidization in excess of $s^*$ causes equilibrium inefficiency.
\begin{proposition}\label{cor_inefficiency}
In symmetric games, subsidies $\hat{s}$ such that $\hat{s}>s^*$ are inefficient. Specifically, the subsidy $\hat{s}$ induces coordination on $p^{\hat{x}}$ where $\hat{x}<\underline{x}$. Players thus coordinate on an inefficient outcome of $\Gamma(x)$ for all $x\in(\hat{x},\underline{x}-\varepsilon/2)$ with probability 1.
\end{proposition}
The possibility of policy-induced coordination failure due to excessive subsidization is illustrated in Figure~\ref{fig_failure}.
\begin{figure}[h]
	\centering
	\includegraphics[width=0.85\textwidth, angle=0, clip=true, trim=0cm 6.5cm 0cm 6.4cm]{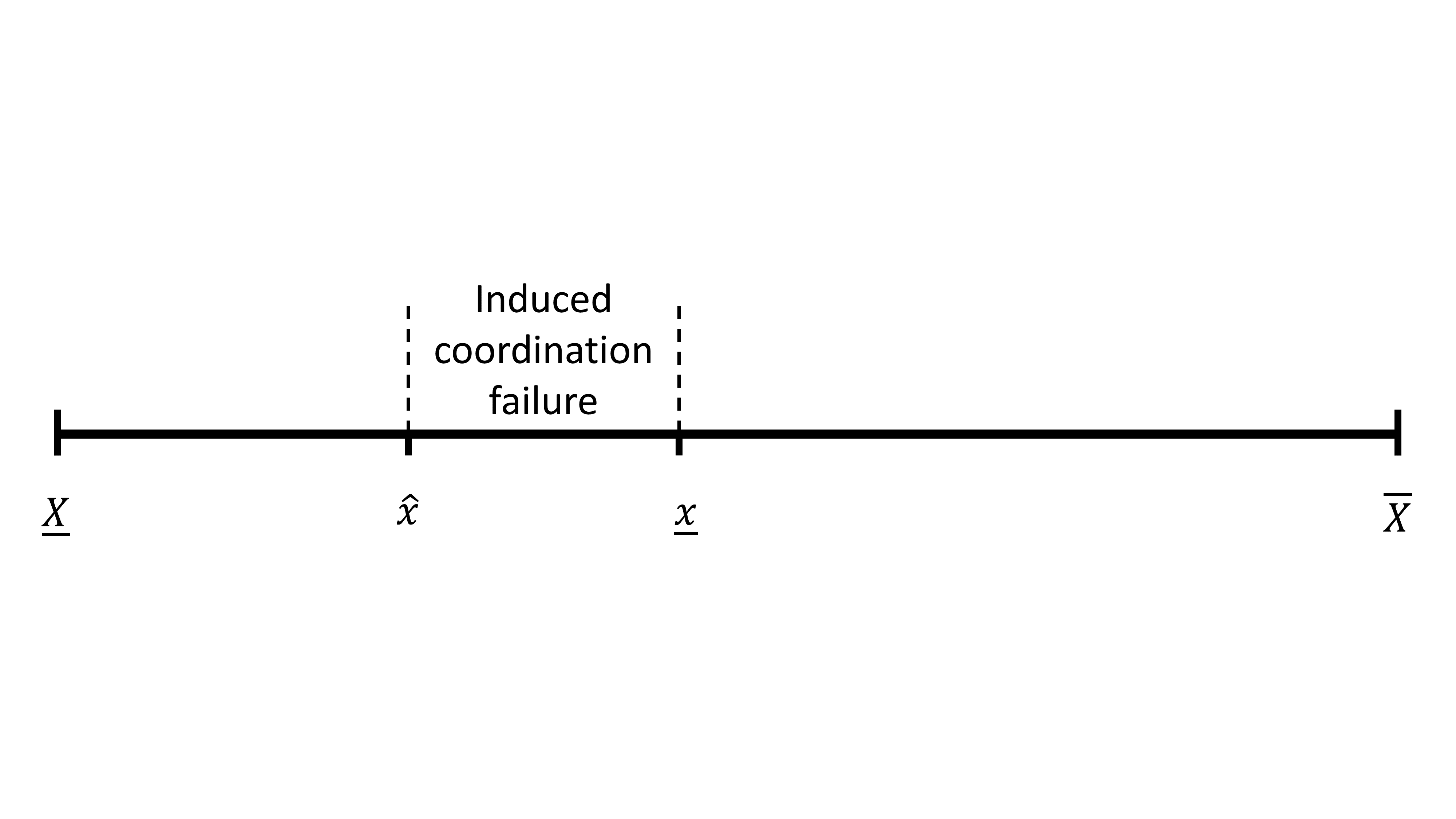}
	\caption{If $\hat{s}>s^*$, the unique equilibrium of $\Gamma^{\varepsilon}(\hat{s})$ has players coordinate on $p^{\hat{x}}$ with $\hat{x}<\underline{x}$. The policy $\hat{s}$ thus induces coordination on an inefficient outcome of the game -- a policy-induced coordination failure -- for all $x\in(\hat{x},\underline{x}-\varepsilon/2)$ with proabbility 1.}
	\label{fig_failure}
\end{figure}
Though Proposition~\ref{cor_inefficiency} is, in some sense, a direct corollary to Theorem~\ref{thm_subsidy}, we single it out to emphasize an important economic implication of our analysis. Prior uncertainty about $x$ implies prior uncertainty about the efficient outcome of $\Gamma(x)$. If the Planner must commit to her policy before $x$ is drawn, the consequent uncertainty about the state $x$ warrants policy moderation as high subsidies risk stimulating coordination on $(1,1,...,1)$ even when $(0,0,...,0)$ turns out to be the ex post efficient outcome. Intuitively, under prior uncertainty about payoff functions the Planner should subsidize conservatively to avoid picking inefficient winners; the Planner should not be wedded to the idea that coordination on $(1,1,...,1)$ must always be achieved.

The notion that policy intervention can itself be a source of coordination failure appears to square well with the historical evidence \citep[\textit{cf.}][]{cowan1990nuclear,cowan1996sprayed}. The theoretical literature has not always emphasized this possibility \citep[a notable exception being][]{cowan1991tortoises}. For example, in the coordination problems studied by \cite{segal2003coordination}, \cite{winter2004incentives}, \cite{angeletos2006signaling}, \cite{sakovics2012matters}, \cite{bernstein2012contracting}, \cite{basak2020diffusing}, \cite{halac2020raising,halac2021rank}, and \cite{kets2022value}, the Pareto-dominant outcome of the game is known a priori. When the efficient outcome of the game is known with certainty, subsidization can be excessive only in the sense that the Planner ends up spending more on subsidies than would be strictly necessary. The results in this paper complement that possibility by making explicit another, non-budgetary source of policy inefficiency: the possibility that excessive intervention induces coordination on an inefficient outcome.

\section{Concluding Remarks}\label{section_concl}
This paper presents a number of results on policy design in coordination games. Strategic uncertainty complicates policy design in coordination games. To deal with this complication, the Planner in this paper connects the problem of policy design to that of equilibrium selection using a global games approach. Our main result characterizes the subsidy scheme that induces coordination on a given equilibrium of the game as its unique equilibrium. We show that optimal subsidies are unique and admit a number of properties that run counter to well-known results on policy design in coordination games. In particular, we show that optimal subsidies are symmetric for identical players, continuous functions of model parameters, and do not make the targeted strategies strictly dominant for any single player.

Two core features of the game considered here help explain the differences between our optimal policy and the policies previously proposed in the literature. First, as stated above, the Planner in this paper connects the problem of policy design to that of equilibrium selection. Equilibrium selection allows the Planner to make very precise inferences about player's actual, rather than hypothetical, strategic beliefs and to design her policy in response to those. Second, the Planner must commit to her policy before knowing which strategy vector will be the ex post efficient outcome of the game. This kind of fundemantal uncertainty leads to a degree of policy restraint as overly agressive intervention may itself become a source of ex post coordination failure.

The analysis also highlights an unraveling effect of policy in coordination games. A subsidy raises a player $i$'s incentive to play the subsidized action. The raised incentive of player $i$ also indirectly increases player $j$'s incentive to play that action. This, in turn, makes the subsidized action even more attractive for player $i$, and so on. Under common knowledge of the policy, this positive feedback loop compounds indefinitely and allows seemingly modest policies to unravel coordination problems.

\appendix
\section{Properties of $\Gamma^\varepsilon$ When $\varepsilon$ Is Small}\label{app_small}

As stated when introducing Corollary~\ref{cor_thm}, the proofs in this Appendix rely upon our ability to analyze the problem ``as if'' the common prior $g$ were uniform when $\varepsilon$ is sufficiently small. Here, we make this claim more precise.

Let us write $\phi^\varepsilon$ for the density of $\varepsilon\cdot\eta_i$. Although in general $\phi^\varepsilon(z)$ can, for any $z$, become arbitrarily large if we pick $\varepsilon$ very small, it remains true that $\phi^\varepsilon(x_i^\varepsilon-x) \prod_{j\neq i}\phi^\varepsilon(x_j^\varepsilon-x) \wrt x$ is (proportional to) a density and, consequently, that for any continuous function $h:\mathcal{X}\to\R$ the quantity $\int h(x) \phi^\varepsilon(x_i^\varepsilon-x)\prod_{j\neq i}\phi^\varepsilon(x_j^\varepsilon-x) \wrt x$ is bounded. In particular, therefore, we know that $\int g(x) \phi^\varepsilon(x_i^\varepsilon-x)\prod_{j\neq i}\phi^\varepsilon(x_j^\varepsilon-x) \wrt x$ is bounded. 

Conditional on his signal $x_i^\varepsilon$ the density of player $i$ on the vector of signals $x_{-i}^\varepsilon$ received by his opponents is
\begin{equation}
	f_i^\varepsilon(x_{-i}^\varepsilon\mid x_i^\varepsilon)=\frac{\int g(x)\phi^\varepsilon(x_i^\varepsilon-x)\prod_{j\neq i}\phi^\varepsilon(x_j^\varepsilon-x)\wrt x}{\iint g(x)\phi^\varepsilon(x_i^\varepsilon-x)\prod_{j\neq i}\phi^\varepsilon(x_j^\varepsilon-x)\wrt x_{-i}^\varepsilon\wrt x}
\end{equation}
for all $x_i^\varepsilon\in[\underline{X}-\varepsilon/2,\overline{X}+\varepsilon/2]$ and all $x_j^\varepsilon\in[x_i^\varepsilon-\varepsilon,x_i^\varepsilon+\varepsilon]$ while $f_i^\varepsilon(x_{-i}^\varepsilon\mid x_i^\varepsilon)=0$ otherwise. Under a uniform prior $g$ the density $f_i^\varepsilon(x_{-i}^\varepsilon\mid x_i^\varepsilon)$ simplifies to $\overline{f_i^\varepsilon}(x_{-i}^\varepsilon\mid x_i^\varepsilon):=\int\phi^\varepsilon(x_i^\varepsilon-x)\prod_{j\neq i}\phi^\varepsilon(x_j^\varepsilon-x)\wrt x$.

\begin{proposition}\label{lemma_uniform_prior}
	For all $\delta>0$, there exists $\varepsilon(\delta)>0$ such that $|f_i^\varepsilon(x_{-i}^\varepsilon\mid x_i^\varepsilon)-\overline{f_i^\varepsilon}(x_{-i}^\varepsilon\mid x_i^\varepsilon)|<\delta$ for all $\varepsilon\leq\varepsilon(\delta)$ and all $(x_i^\varepsilon,x_{-i}^\varepsilon)\in \R^N$.
\end{proposition}
\begin{proof}
	Given $x_i^\varepsilon$, for all $x\in[x_i^\varepsilon-\varepsilon/2,x_i^\varepsilon+\varepsilon/2]$, let us define $g^\varepsilon_-(x_i^\varepsilon)=\min_{x}g(x)$ and $g^\varepsilon_+(x_i^\varepsilon)=\max_{x}g(x)$. Clearly, $g^\varepsilon_-(x_i^\varepsilon)\leq g(x)\leq g^\varepsilon_+(x_i^\varepsilon)$ in the relevant domain. Therefore
	\begin{align*}
		\frac{\int g^\varepsilon_-(x_i^\varepsilon)\phi^\varepsilon(x_i^\varepsilon-x)\prod_{j\neq i}\phi^\varepsilon(x_j^\varepsilon-x)\wrt x}{\iint g^\varepsilon_+(x_i^\varepsilon)\phi^\varepsilon(x_i^\varepsilon-x)\prod_{j\neq i}\phi^\varepsilon(x_j^\varepsilon-x)\wrt x_{-i}^\varepsilon\wrt x}&\leq f_i^\varepsilon(x_{-i}^\varepsilon\mid x_i^\varepsilon)\\
		&\leq \frac{\int g^\varepsilon_+(x_i^\varepsilon)\phi^\varepsilon(x_i^\varepsilon-x)\prod_{j\neq i}\phi^\varepsilon(x_j^\varepsilon-x)\wrt x}{\iint g^\varepsilon_-(x_i^\varepsilon)\phi^\varepsilon(x_i^\varepsilon-x)\prod_{j\neq i}\phi^\varepsilon(x_j^\varepsilon-x)\wrt x_{-i}^\varepsilon\wrt x}
	\end{align*}
	for all $(x_i^\varepsilon,x_{-i}^\varepsilon)\in\R^N$ and all $\varepsilon>0$. Because $g^\varepsilon_-(x_i^\varepsilon)$ and $g^\varepsilon_+(x_i^\varepsilon)$ are constants relative to the variable of integration, we can factor them out of the integral. Noting that $\iint \phi^\varepsilon(x_i^\varepsilon-x)\prod_{j\neq i}\phi^\varepsilon(x_j^\varepsilon-x)\wrt x_{-i}^\varepsilon\wrt x=1$, the above then becomes
	\begin{align*}
		\frac{g^\varepsilon_-(x_i^\varepsilon)}{g^\varepsilon_+(x_i^\varepsilon)}\int \phi^\varepsilon(x_i^\varepsilon-x)\prod_{j\neq i}\phi^\varepsilon(x_j^\varepsilon-x)\wrt x\leq f_i^\varepsilon(x_{-i}^\varepsilon\mid x_i^\varepsilon)\leq \frac{g^\varepsilon_+(x_i^\varepsilon)}{g^\varepsilon_-(x_i^\varepsilon)}\int \phi^\varepsilon(x_i^\varepsilon-x)\prod_{j\neq i}\phi^\varepsilon(x_j^\varepsilon-x)\wrt x,
	\end{align*}
	or
	\begin{align*}
		\frac{g^\varepsilon_-(x_i^\varepsilon)}{g^\varepsilon_+(x_i^\varepsilon)}\,{\overline{f_i^\varepsilon}}(x_{-i}^\varepsilon\mid x_i^\varepsilon)\leq f_i^\varepsilon(x_{-i}^\varepsilon\mid x_i^\varepsilon)\leq \frac{g^\varepsilon_+(x_i^\varepsilon)}{g^\varepsilon_-(x_i^\varepsilon)}\, {\overline{f_i^\varepsilon}}(x_{-i}^\varepsilon\mid x_i^\varepsilon).
	\end{align*}
	From the uniform continuity of $g$ (i.e.\ $g$ is continuous on a compact set, which by the Heine-Cantor theorem implies $g$ is uniformly continuous) follows that for any $k>0$ there exists $\varepsilon(k)>0$ such that $g^\varepsilon_+(x_i^\varepsilon)-g^\varepsilon_-(x_i^\varepsilon)<k$ for all $\varepsilon\leq\varepsilon(k)$ and all $x_i^\varepsilon$. It follows immediately that for all $\delta>0$ there exists $\varepsilon(\delta)>0$ such that $|{\overline{f_i^\varepsilon}}(x_{-i}^\varepsilon\mid x_i^\varepsilon)- <f_i^\varepsilon(x_{-i}^\varepsilon\mid x_i^\varepsilon)|<\delta$ for all $\varepsilon\leq\varepsilon(\delta)$ and all $(x_i^\varepsilon,x_{-i}^\varepsilon)\in \R^N$.
\end{proof}
%

An immediate implication of Proposition~\ref{lemma_uniform_prior} is that the cumulative distribution function $F(z_{-i}\mid x_i^\varepsilon):=\int^{z_{-i}}f_i^\varepsilon(x_{-i}^\varepsilon\mid x_i^\varepsilon)\wrt x_{-i}^\varepsilon$ can also, for sufficiently small $\varepsilon$, be approximated arbitrarily closely by the distribution $\overline{F_i^\varepsilon}$ obtained under a uniform prior $g$. Moreover, the probability distribution $\overline{F_i^\varepsilon}$ admits a highly useful property: its shape is independent of $x_i^\varepsilon$. To be more precise, and abusing notation, let us write $\Delta$ for both a real number $\Delta\in\R$ and the vector of real numbers $(\Delta,\Delta,...,\Delta)\in\R^{N-1}$ such that $z_{-i}+\Delta=(z_j+\Delta)_{j\neq i}$.

\begin{proposition}\label{lemma_shape}
	For all $\Delta$ and all $(z_i,z_{-i})\in\R^N$, we have $\overline{F_i^\varepsilon}(z_{-i}+\Delta\mid z_i+\Delta)=\overline{F_i^\varepsilon}(z_{-i}\mid z_i)$.
\end{proposition}
\begin{proof}
	Fix $(z_i,z_{-i})\in\R^N$ and $\Delta$. We have
	\begin{align*}
		{\overline{F_i^\varepsilon}}(z_{-i}\mid z_i)&=\int\limits_{z_i-\varepsilon}^{z_{-i}}{\overline{f_i^\varepsilon}}(x_{-i}^\varepsilon\mid z_i)\wrt x_{-i}^\varepsilon\\
		&=\int\limits_{z_i-\varepsilon}^{z_{-i}} \left[\int\limits_{z_i-\varepsilon/2}^{z_i+\varepsilon/2}\phi^\varepsilon(z_i-x)\prod_{j\neq i}\phi^\varepsilon(x_j^\varepsilon-x)\wrt x\right]\wrt x_{-i}^\varepsilon\\
		&=\int\limits_{z_i-\varepsilon}^{z_{-i}} \left[\int\limits_{z_i-\varepsilon/2}^{z_i+\varepsilon/2}\phi^\varepsilon(z_i+\Delta-(x+\Delta))\prod_{j\neq i}\phi^\varepsilon(x_j^\varepsilon+\Delta-(x+\Delta))\wrt x\right]\wrt x_{-i}^\varepsilon\\
		&=\int\limits_{z_i-\varepsilon}^{z_{-i}} \left[\int\limits_{z_i+\Delta-\varepsilon/2}^{z_i+\Delta+\varepsilon/2}\phi^\varepsilon(z_i+\Delta-x')\prod_{j\neq i}\phi^\varepsilon(x_j^\varepsilon+\Delta-x')\wrt x'\right]\wrt x_{-i}^\varepsilon\\
		&=\int\limits_{z_i-\varepsilon}^{z_{-i}}{\overline{f_i^\varepsilon}}(x_{-i}^\varepsilon+\Delta\mid z_i+\Delta)\wrt x_{-i}^\varepsilon\\
		&=\int\limits_{z_i+\Delta-\varepsilon}^{z_{-i}+\Delta}{\overline{f_i^\varepsilon}}(x_{-i}^\varepsilon\mid z_i+\Delta)\wrt x_{-i}^\varepsilon\\
		&={\overline{F_i^\varepsilon}(z_{-i}+\Delta\mid z_i+\Delta)},
	\end{align*}
as claimed.
\end{proof}

\section{Proofs}

Let $h^\varepsilon$ denote a function that is (implicitly) parametrized by $\varepsilon$, and let $H$ be defined on the same domain as $h^\varepsilon$. Throughout this Appendix, when we write $h^\varepsilon(z)\to H(z)$ we mean that for all $\delta>0$ there exists $\varepsilon(\delta)>0$ such that $|h^\varepsilon(z)-H(z)|<\delta$ for all $\varepsilon\leq\varepsilon(\delta)$ and all $z$ in the domain of $h^\varepsilon$ and $H$. Thus, $h^\varepsilon(z)\to H(z)$ should be read as saying that $h^\varepsilon(z)$ can be brought arbitrarily close to $H(z)$ provided we choose $\varepsilon$ sufficiently small. Whenever such a claim is made without further explanation, it is implied that this follows Proposition~\ref{lemma_uniform_prior}. We emphasize that the symbol ``$\rightarrow$'' should \textit{not} be read as a limit as $\varepsilon$ goes to zero; since $\varepsilon>0$ by assumption, that limit is not defined.

\vspace{5mm}

\noindent PROOF OF LEMMA~\ref{lemma_increasing}
\begin{proof}
	First, observe that
	\begin{align*}
		u_i^\varepsilon(p_{-i}\mid x_i^\varepsilon)&=\int u_i\left(p_{-i}\left(x_{-i}^\varepsilon\right)\mid x\right)\, \wrt F_i^\varepsilon(x,x_{-i}^\varepsilon\mid x_i^\varepsilon)\\
		&=\int w_i\left(p_{-i}\left(x_{-i}^\varepsilon\right)\right)+x\,\wrt F_i^\varepsilon(x,x_{-i}^\varepsilon\mid x_i^\varepsilon)-c_i\\
		&\to\int w_i\left(p_{-i}\left(x_{-i}^\varepsilon\right)\right)\,\wrt {\overline{F_i^\varepsilon}}(x_{-i}^\varepsilon\mid x_i^\varepsilon)+x_i^\varepsilon-c_i,
	\end{align*}
	for any strategy vector $p_{-i}$.
	
	To prove part (i), it suffices to show that $\int w_i\left(p_{-i}^y\left(x_{-i}^\varepsilon\right)\right)\,\wrt {\overline{F_i^\varepsilon}}(x_{-i}^\varepsilon\mid x_i^\varepsilon)$ is increasing in $x_i^\varepsilon$. First we introduce a random variable $v_i(x_{-i})=w_i(p_{-i}^y(x_{-i}^\varepsilon))$ and observe that, since $w_i(p_{-i}^y(x_{-i}^\varepsilon))$ is increasing in $p_{-i}^y(x_{-i}^\varepsilon)$ and $p_{-i}^y(x_{-i}^\varepsilon)$ is increasing in $x_{-i}^\varepsilon$, $v_i$ is increasing in $x_{-i}^\varepsilon$. Next, we note that the distribution ${\overline{F_i^\varepsilon}}(x_{-i}^\varepsilon\mid x_i^\varepsilon)$ is first-order stochastic dominant over the distribution $F_i^\varepsilon(x_{-i}^\varepsilon\mid \hat{x}_i^\varepsilon)$ iff $x_i^\varepsilon>\hat{x}_i^\varepsilon$; this follows from Bayes' theorem upon application of the two facts that (a) each $\varepsilon_j$ (and indeed $\varepsilon_i$) is drawn independently of $x$, and (b) player $i$'s conditional distribution on $x$ given $x_i^\varepsilon$ first-order stochastic dominates his conditional distribution on $x$ given $\hat{x}_i^\varepsilon$ iff $x_i^\varepsilon>\hat{x}_i^\varepsilon$. Hence, because $v_i$ is increasing we have $\int v_i(x_{-i}^\varepsilon)\wrt {\overline{F_i^\varepsilon}}(x_{-i}^\varepsilon\mid x_i^\varepsilon)>\int v_i(x_{-i}^\varepsilon)\wrt F_i^\varepsilon(x_{-i}^\varepsilon\mid \hat{x}_i^\varepsilon)$ and the result follows.
	
	To prove part (ii), we reiterate the observation from the proof of part (i) that the distribution ${\overline{F_i^\varepsilon}}(x_{-i}^\varepsilon\mid x_i^\varepsilon)$ is first-order stochastic dominant over the distribution ${\overline{F_i^\varepsilon}}(x_{-i}^\varepsilon\mid \hat{x}_i^\varepsilon)$ iff $x_i^\varepsilon>\hat{x}_i^\varepsilon$. Next, we note that $p_{-i}^y(x_{-i}^\varepsilon)$ is (weakly) decreasing in $y_j\in y$, all $j\neq i$ (and, therefore, the random variable $v_i(x_{-i}^\varepsilon)$ we introduced in the proof of part (i) is also decreasing in $y_j$). Therefore $\int w_i\left(p_{-i}^y\left(x_{-i}^\varepsilon\right)\right)\,\wrt {\overline{F_i^\varepsilon}}(x_{-i}^\varepsilon\mid x_i^\varepsilon)$ is decreasing in $y_j$ and the result follows.
\end{proof}
%
\noindent PROOF OF LEMMA~\ref{lemma_unique_limit}
\begin{proof}
	We omit the argument $s$ to reduce notation. By construction, $l_i\leq r_i$. Define $\Delta_i:=r_i-l_i$, so $\Delta_i\geq0$. We first establish a useful claim.
	
	\begin{claim}
		If $\Delta_i=\Delta$ for all $i\in \N$, then $\Delta=0$.
	\end{claim}
	\begin{proof}[Proof of the claim]
		If $\Delta_i=\Delta$ for all $i\in\N$, we have
		\begin{align*}
			u_i^\varepsilon(p_{-i}^{r_{-i}}\mid r_i,s_i)&\to r_i+\int w_i(p_{-i}^{r_{-i}}(x_{-i}^\varepsilon))\wrt {\overline{F_i^\varepsilon}}(x_{-i}^\varepsilon\mid r_i)+s_i-c_i\\
			&=l_i+\Delta
			+w_i(p_{-i}^{l_{-i}+\Delta}(x_{-i}^\varepsilon))\wrt {\overline{F_i^\varepsilon}}(x_{-i}^\varepsilon\mid l_i+\Delta)+s_i-c_i\\
			&=  l_i+\Delta+\int w_i(p_{-i}^{l_{-i}}(x_{-i}^\varepsilon))\wrt {\overline{F_i^\varepsilon}}(x_{-i}^\varepsilon\mid l_i)+s_i-c_i\\
			&\to\Delta+u_i^\varepsilon(p_{-i}^{l_{-i}}\mid l_i,s_i).
		\end{align*}
		By construction, $u_i^\varepsilon(p_{-i}^{r_{-i}}\mid r_i,s_i)=u_i^\varepsilon(p_{-i}^{l_{-i}}\mid l_i,s_i)$, and it follows that $\Delta=0$.
	\end{proof}	
	Now let $\Delta_i\neq\Delta_j$ for at least one pair of players $i,j\in\N$ and suppose (without loss) that player $i$ is such that $\Delta_i=\max\{\Delta_j\mid j\in\N\}$. Because $\Delta_i\geq\Delta_j$ for all $j\neq i$ with a strict inequality for at least one $j$, we have
	\begin{align*}
		&u_i^\varepsilon(p_{-i}^{r_{-i}}\mid r_i,s_i)-u_i^\varepsilon(p_{-i}^{l_{-i}}\mid l_i,s_i)\\
		&\to r_i-l_i+\int w_i(p_{-i}^{r_{-i}}(x_{-i}^\varepsilon))\wrt {\overline{F_i^\varepsilon}}(x_{-i}^\varepsilon\mid r_i)-\int w_i(p_{-i}^{l_{-i}}(x_{-i}^\varepsilon))\wrt {\overline{F_i^\varepsilon}}(x_{-i}^\varepsilon\mid l_i)\\
		&=\Delta_i+\int w_i(p_{-i}^{r_{-i}}(x_{-i}^\varepsilon))\wrt {\overline{F_i^\varepsilon}}(x_{-i}^\varepsilon\mid r_i)-\int w_i(p_{-i}^{l_{-i}}(x_{-i}^\varepsilon))\wrt {\overline{F_i^\varepsilon}}(x_{-i}^\varepsilon\mid l_i)\\
		&>\int w_i(p_{-i}^{r_{-i}}(x_{-i}^\varepsilon))\wrt {\overline{F_i^\varepsilon}}(x_{-i}^\varepsilon\mid r_i)-\int w_i(p_{-i}^{l_{-i}}(x_{-i}^\varepsilon))\wrt {\overline{F_i^\varepsilon}}(x_{-i}^\varepsilon\mid l_i)\\
		&>\int w_i(p_{-i}^{l_{-i}+\Delta_i}(x_{-i}^\varepsilon))\wrt {\overline{F_i^\varepsilon}}(x_{-i}^\varepsilon\mid l_i+\Delta_i)-\int w_i(p_{-i}^{l_{-i}}(x_{-i}^\varepsilon))\wrt {\overline{F_i^\varepsilon}}(x_{-i}^\varepsilon\mid l_i)\\
		&=0,
	\end{align*}
	where the first inequality follows from $\Delta_i>0$ and the final equality is a consequence of Property~\ref{lemma_shape}. Hence, for player $i$ we have $u_i^\varepsilon(p_{-i}^{r_{-i}}\mid r_i,s_i)>u_i^\varepsilon(p_{-i}^{l_{-i}}\mid l_i,s_i)$, contradicting that $u_i^\varepsilon(p_{-i}^{r_{-i}}\mid r_i,s_i)=u_i^\varepsilon(p_{-i}^{l_{-i}}\mid l_i,s_i)$ by construction. Hence, there cannot be a player $i$ such that $\Delta_i\geq\Delta_j$ for all $j\neq i$ with a strict inequality for at least one $j$. Therefore $\Delta_i=\Delta$ for all $i\in\N$. By the claim at the start of the this proof, this implies $\Delta=0$.
\end{proof}

\noindent PROOF OF LEMMA~\ref{lemma_unique_subsidy}
\begin{proof}
	Suppose, in contrast, that there are two distinct vectors of subsidies $\hat{s}_1=(\hat{s}_{1i})$ and $\hat{s}_2=(\hat{s}_{2i})$ that both implement $p^{\hat{x}}$ such that $\hat{s}_1\neq\hat{s}_2$. Per Lemmas~\ref{lemma_unique_limit} and \ref{thm_BNE}, $\hat{s}_1$ and $\hat{s}_2$ must solve $x(\hat{s}_1)=x(\hat{s}_2)=\hat{x}$. By \eqref{eq_x(s)}, this means that $\hat{s}_{1i}$ and $\hat{s}_{2i}$ are both solutions to
	\begin{equation}
		u_i^\varepsilon\left(p_{-i}^{\hat{x}}\mid \hat{x}_i,\hat{s}_{1i}\right)=u_i^\varepsilon\left(p_{-i}^{\hat{x}}\mid \hat{x}_i,\hat{s}_{2i}\right)=0,
	\end{equation} 
	for each $i\in\N$. Using \eqref{eq_exp_incentive_sub}, we thus have
	\begin{equation}
		u_i^\varepsilon\left(p_{-i}^{\hat{x}}\mid \hat{x}_i\right)+s_{1i}=u_i^\varepsilon\left(p_{-i}^{\hat{x}}\mid \hat{x}_i\right)+s_{2i},
	\end{equation}
	which implies
	\begin{equation}
		\hat{s}_{1i}=\hat{s}_{2i}
	\end{equation}
	for all $i\in\N$. This contradicts our assumption that $\hat{s}_1\neq\hat{s}_2$.
\end{proof}

\noindent PROOF OF LEMMA~\ref{thm_BNE}
\begin{proof}
	Let $p=(p_i)$ be a BNE of $\Gamma^{\varepsilon}(s)$. For any player $i$, define
	\begin{equation}
		\underline{\underline{x}}_i=\inf\{x_i^{\varepsilon}\mid p_i(x_i^{\varepsilon})>0\},
	\end{equation}
	and
	\begin{equation}
		\overline{\overline{x}}_i=\sup\{x_i^{\varepsilon}\mid p_i(x_i^{\varepsilon})<1\}.
	\end{equation}
	Observe that $\underline{\underline{x}}_i\leq \overline{\overline{x}}_i$. Now define 
	\begin{equation}
		\underline{\underline{x}}=\min \{\underline{\underline{x}}_i\},
	\end{equation}
	and
	\begin{equation}
		\overline{\overline{x}}=\max \{ \overline{\overline{x}}_i\}.
	\end{equation}
	By construction, $\overline{\overline{x}}\geq  \overline{\overline{x}}_i\geq\underline{\underline{x}}_i\geq\underline{\underline{x}}$. Observe that $p$ is a BNE of $\Gamma^{\varepsilon}(s)$ only if, for each $i$, it holds that $u_i^{\varepsilon}(p_{-i}(x_{-i}^{\varepsilon})\mid \underline{\underline{x}}_i)\geq0$. Consider then the expected incentive $u_i^{\varepsilon}(p_{-i}^{\underline{\underline{x}}}(x_{-i}^{\varepsilon})\mid \underline{\underline{x}}_i)$. It follows from the definition of $\underline{\underline{x}}$ that $p^{\underline{\underline{x}}}(x^{\varepsilon})\geq p(x^{\varepsilon})$ for all $x^{\varepsilon}$. The implication is that, for each $i$, $u_i^{\varepsilon}(p_{-i}^{\underline{\underline{x}}}(x_i{-i}^{\varepsilon})\mid \underline{\underline{x}}_i)\geq u_i^{\varepsilon}(p_{-i}(x_{-i}^{\varepsilon})\mid \underline{\underline{x}}_i)\geq0$. From Proposition~\ref{lemma_monotone} then follows that $\underline{\underline{x}}\geq x $. 
	
	Similarly, if $p$ is a BNE of $\Gamma^{\varepsilon}(s)$ then, for each $i$, it must hold that $u_i^{\varepsilon}(p_{-i}(x_{-i}^{\varepsilon})\mid  \overline{x}_i)\leq0$. Consider the expected incentive $u_i^{\varepsilon}(p_{-i}^{ \overline{x}}(x_{-i}^{\varepsilon})\mid  \overline{\overline{x}}_i)$. It follows from the definition of $\overline{\overline{x}}$ that $p^{\overline{\overline{x}}}(x^{\varepsilon})\leq p(x^{\varepsilon})$ for all $x^{\varepsilon}$. For each $i$ it therefore holds that $u_i^{\varepsilon}(p_{-i}^{\overline{\overline{x}}}(x_{-i}^{\varepsilon})\mid  \overline{\overline{x}}_i)\leq u_i^{\varepsilon}(p_{-i}(x_i{-i}^{\varepsilon})\mid  \overline{\overline{x}}_i)\leq0$. Hence $\overline{\overline{x}}\leq x $.
	
	Since $\underline{\underline{x}}\leq \overline{\overline{x}}$ while also $\underline{\underline{x}}\geq x $ and $\overline{\overline{x}}\leq x $ it must hold that $\underline{\underline{x}}=\overline{\overline{x}}=x $. Moreover, since $p^{\underline{\underline{x}}}\geq p$ while also $p^{\overline{\overline{x}}}\leq p$, given $\underline{\underline{x}}=\overline{\overline{x}}=x $, it follows that $p_i(s_i^{\varepsilon})=p_i^{x }(x_i^{\varepsilon})$ for all $x_i^{\varepsilon}\neq x $ and all $i$ (recall that for each player $i$ one has $u_i^{\varepsilon}(p^{x }_{-i}\mid x )=0$, explaining the singleton exeption at $x_i^{\varepsilon}=x $). Thus, if $p=(p_i)$ is a BNE of $\Gamma^{\varepsilon}(s)$ then it must hold that $p_i(x_i^{\varepsilon})=p_i^{x }(x_i^{\varepsilon})$ for all $x_i^{\varepsilon}\neq x $ and all $i$, as we needed to prove.
\end{proof}

\noindent PROOF OF LEMMA~\ref{lemma_monotone}
\begin{proof}
	Let $\Omega^\varepsilon(n\mid X,x_i^\varepsilon)$ denote the probability that a player $i$ who observes signal $x_i^\varepsilon$ attaches to the event that $n$ other players $j$ receive a signal $x_j^\varepsilon\geq X$:
	\begin{equation}
		\Omega^\varepsilon(n\mid X,x_i^\varepsilon)=\frac{\int g(x)f\phi^\varepsilon\left(x_i^\varepsilon-x\right)\binom{N-1}{n}\left[\Phi\left(X-x\right)\right]^{N-n-1}\left[1-\Phi\left(X-x\right)\right]^n\wrt x}{\int g(x)\phi^\varepsilon\left(x_i^\varepsilon-x\right)\wrt x},
	\end{equation}
	where $\Phi^\varepsilon(z):=\int_{\varepsilon/2}^z\phi^\varepsilon(\lambda)\wrt \lambda$ is the c.d.f.\ of $\phi^\varepsilon$. When $g$ is uniform, this simplifies to:
	\begin{equation}
		{\overline{\Omega^\varepsilon}}(n\mid X,x_i^\varepsilon)=\binom{N-1}{n}\int \phi^\varepsilon\left(x_i^\varepsilon-x\right)\left[\Phi^\varepsilon\left(X-x\right)\right]^{N-n-1}\left[1-\Phi^\varepsilon\left(X-x\right)\right]^n\wrt x
	\end{equation}

	Clearly, if player $i$'s opponents play $p_{-i}^{X}$ then $\Omega^\varepsilon(n\mid X,x_i^\varepsilon)$ is also $i$'s conditional distribution on $\sum_{j\neq i}a_j=n$. Therefore
	\begin{align*}
		u_i^\varepsilon(p_{-i}^X\mid x_i^\varepsilon,s_i)&=\int x\phi^\varepsilon(x_i^\varepsilon-x)\wrt x+\sum_{n=0}^{N-1}w_i(n)\Omega^\varepsilon(n\mid X,x_i^\varepsilon)-c_i+s_i\\
		&\to x_i^\varepsilon+\sum_{n=0}^{N-1}w_i(n){\overline{\Omega^\varepsilon}}(n\mid X,x_i^\varepsilon)-c_i+s_i
	\end{align*}
	as $\varepsilon\to0$. To prove the Lemma, we need only evaluate $\overline{\Omega^\varepsilon}(n\mid X,x_i^\varepsilon)$ at $x_i^\varepsilon=X$. Define $y:=X-x$, so we may write $\int \phi^\varepsilon\left(y\right)\left[\Phi^\varepsilon\left(y\right)\right]^{N-n-1}\left[1-\Phi^\varepsilon\left(y\right)\right]^n\wrt y$.\footnote{To evaluate this integral, recall that for two functions $u$ and $v$ of $y$ intergration by parts gives
	\[\int_a^bu(y)v'(y)dy=[u(y)v(y)]_a^b-\int_a^bu'(y)v(y)dy.\]
	A convenient choice of $u$ and $v$ will prove to be $v'(y):=\phi^\varepsilon(y)[\Phi^\varepsilon(y)]^{N-n-1}$ and $u(y):=[1-\Phi^\varepsilon(y)]^n$. We thus have $u'(y)=-n[1-\Phi^\varepsilon(y)]^{n-1}\phi^\varepsilon(y)$ and $v(y)=\frac{1}{N-n}[\Phi^\varepsilon(y)]^{N-n}$.} Repeatedly carrying out the integration by parts, we obtain
	\begin{align*}
		\frac{1}{\binom{N-1}{n}}\cdot{\overline{\Omega^\varepsilon}}(n\mid X,X)&=\int \phi^\varepsilon(y)\left[\Phi^\varepsilon(y)\right]^{N-n-1}\left[1-\Phi^\varepsilon(y)\right]^n\wrt y\\
		&=\frac{n}{N-n}\int \phi^\varepsilon(y)\left[\Phi^\varepsilon(y)\right]^{N-n}\left[1-\Phi^\varepsilon(y)\right]^{n-1}\wrt y=\\
		&=\frac{n\cdot(n-1)}{(N-n)\cdot(N-n+1)} \int \phi^\varepsilon(y)\left[\Phi^\varepsilon(y)\right]^{N-n+1}\left[1-\Phi^\varepsilon(y)\right]^{n-2}\wrt y\\
		&\vdots\\
		&=\frac{n\cdot (n-1)\cdot (n-2)\cdots 1}{(N-n)\cdot (N-n+1)\cdots (N-1)}\int \phi^\varepsilon(y)[\Phi^\varepsilon(y)]^{N-1}dy\\
		&=\frac{n!(N-n-1)!}{(N-1)!}\frac{1}{N}[\Phi^\varepsilon(y)]_{-\infty}^{\infty}\\
		&=\frac{1}{N}\frac{1}{\binom{N-1}{n}},
	\end{align*}
	which shows that $\overline{\Omega^\varepsilon}(n\mid X,X)=1/N$ for all $n=0,1,...,N-1$. Therefore
	\begin{align*}
		u_i^\varepsilon(p_{-i}^X\mid X,s_i)\to X+\sum_{n=0}^{N-1}w_i(n){\overline{\Omega^\varepsilon}}(n\mid X,X)-c_i+s_i= X+\sum_{n=0}^{N-1}\frac{w_i(n)}{N}-c_i+s_i,
	\end{align*}
	as given.
\end{proof}

\noindent PROOF OF PROPOSITION~\ref{cor_pa}

\begin{proof}
	Given $a_{-i}$ and the reward scheme $v$, the payoff to agent $i$ is given by:
	\begin{equation}
		\pi_i(a_i,a_{-i}\mid x,v_i)=
		\begin{cases}
			v_i-c_i\quad&\text{if the project succeeds and }a_i=1\\ 
			v_i-x\quad&\text{if the project succeeds and }a_i=0\\
			-c_i\quad&\text{if the project does not succeed and }a_i=1\\  
			-x\quad&\text{if the project does not succeed and }a_i=0\\ 
		\end{cases}
	\end{equation}
	Since project success is stochastic and agents do not observe $x$, their (conditional) expected payoff is:
	\begin{equation}
		\pi_i^\varepsilon(a_i,a_{-i}\mid x_i^\varepsilon,v_i)=
		\begin{cases}
			q(A_{-i}+1)\cdot v_i-c_i\quad&\text{if }a_i=1\\ 
			q(A_{-i})\cdot v_i-x_i^\varepsilon\quad&\text{if }a_i=0,
		\end{cases}
	\end{equation}
	yielding her expected incentive to work:
	\begin{equation}
		u_i^\varepsilon(a_{-i}\mid x_i^\varepsilon,v_i)=(q(A_{-i}+1)-q(A_{-i}))\cdot v_i-c_i+x_i^\varepsilon.
	\end{equation}
	The Planner seeks to implement $p^{\tilde{x}}$. The bonus scheme $\tilde{v}$ implements $p^{\tilde{x}}$ iff
	\begin{equation}
		u_i^\varepsilon(p_{-i}^{\tilde{x}}\mid \tilde{x},\tilde{v}_i)=\tilde{v}_i\int\left(q\left(p_{-i}^{\tilde{x}}(x_{-i}^\varepsilon)+1\right)-q\left(p_{-i}^{\tilde{x}}(x_{-i}^\varepsilon)\right)\right)\,\wrt {\overline{F_i^\varepsilon}}(x_{-i}^\varepsilon\mid\tilde{x})-c_i+\tilde{x}=0,
	\end{equation}
	for all $i\in\N$. Invoking Lemma~\ref{lemma_monotone}, we know that
	\[\int\left(q\left(p_{-i}^{\tilde{x}}(x_{-i}^\varepsilon)+1\right)-q\left(p_{-i}^{\tilde{x}}(x_{-i}^\varepsilon)\right)\right)\,\wrt {\overline{F_i^\varepsilon}}(x_{-i}^\varepsilon\mid\tilde{x})=\sum_{n=0}^{N-1}\frac{q(n+1)-q(n)}{N}:=\overline{q}.\]
	Therefore, $\tilde{v}_i$ solves
	\[\overline{q}\cdot \tilde{v}_i-c_i+\tilde{x}=0\implies\overline{q}\cdot\tilde{v_i}=c_i-\tilde{x},\]
	as given.
\end{proof}

%
%

\noindent PROOF OF PROPOSITION~\ref{cor_het}
\begin{proof}
	Recall from the proof of Lemma~\ref{lemma_monotone} that $\Omega^\varepsilon(n\mid X,x_i^\varepsilon)$ denotes the probability that $n$ players $j$ receive a signal $x_j^\varepsilon\geq X$ while $N-n-1$ receive a signal $x_j^\varepsilon<X$. Moreover, given that $n$ players $j$ receive a signal $x_j^\varepsilon$, the probability that any given subset of players $\{j_1,j_2,...,j_n\}\subseteq\N\setminus\{i\}$ receive signals above $X$ is the same (e.g.\ uniform) across such subsets; as there are exactly $\binom{N-1}{n}$ (unique) subsets  $\{j_1,j_2,...,j_n\}\subseteq\N\setminus\{i\}$, this (conditional) probability is simply $1/\binom{N-1}{n}$. Given the strategy vector $p_{-i}^X$ played, and conditional on exactly $n$ players $j$ receiving a signal $x_j^\varepsilon>X$, the expected spillover on player $i$ is hence $\sum_{a_{-i}\in A_{-i}^n}w_i(a_{-i})/\binom{N-1}{n}$, where we recall that $A_{-i}^n:=\{a_{-i}\mid \sum_{a_j\in a_{-i}}a_j=n\}$. Putting all this together, we get
	\begin{align*}
		u_i^\varepsilon(p_{-i}^X\mid x_i^\varepsilon,s_i)&=\frac{\int xg(x)\phi^\varepsilon(x_i^\varepsilon-x)\wrt x}{\int g(x)\phi^\varepsilon(x_i^\varepsilon-x)\wrt x}+\sum_{n=0}^{N-1}\frac{\sum_{a_{-i}\in A_{-i}^n}w_i(a_{-i})}{\binom{N-1}{n}}\Omega^\varepsilon(n\mid X,x_i^\varepsilon)-c_i+s_i\\
		&=\frac{\int xg(x)\phi^\varepsilon(x_i^\varepsilon-x)\wrt x}{\int g(x)\phi^\varepsilon(x_i^\varepsilon-x)\wrt x}+\sum_{n=0}^{N-1}w_i^n\Omega^\varepsilon(n\mid X,x_i^\varepsilon)-c_i+s_i\\
		&\to x_i^\varepsilon+\sum_{n=0}^{N-1}w_i^n{\overline{\Omega^\varepsilon}}(n\mid X,x_i^\varepsilon)-c_i+s_i.
	\end{align*}
	Furthermore, we need only concern ourselves with the event that $x_i^\varepsilon=X$, in which case we have
	\[u_i^\varepsilon(p_{-i}^X\mid X,s_i)\to X+\sum_{n=0}^{N-1}\frac{w_i^n}{N}-c_i+s_i,\]
	where we use the result that $\overline{\Omega^\varepsilon}(n\mid X,X)=1/N$ for all $n=0,1,...,N-1$ established in the proof of Lemma~\ref{lemma_monotone}. Finally, solving $u_i^\varepsilon(p_{-i}^{\tilde{x}}\mid \tilde{x},\tilde{s}_i)=0$ for $\tilde{s}_i$ yields the result.
\end{proof}

\noindent PROOF OF PROPOSITION~\ref{prop_asym}
\begin{proof}
	We prove a more general case in which the player set $\N$ is partitioned into $K$ subsets $\N_k$, $k\in\{1,2,...,K\}$.
	For each $k$, let $\Omega_k^\varepsilon(n_k\mid X_k,x_i^\varepsilon)$ denote the probability that $n_k$ players $j\neq i$ in $\N_k$ receive a signal $x_j^\varepsilon\geq X_k$. We relabel groups so that $X_1<X_2<\hdots<X_K$; moreover, we assume that $X_{k+1}-X_k>\varepsilon$ for all $k$. Note that conditional on his signal $x_i^\varepsilon$, player $i$ knows that $x_j^\varepsilon\in[x_i^\varepsilon-\varepsilon,x_i^\varepsilon+\varepsilon]$ for each $j\neq i$. Hence, for each $k\in\{1,2,...,K\}$ we have $\Omega^\varepsilon_l(N_l\mid X_l, X_k)=1$ for all $l=1,2,...,k-1$ and $\Omega^\varepsilon_l(N_l\mid X_l, X_k)=0$ for all $l=k+1,k+2,...,K$. Moreover, for player $i\in\N_k$ we know that, following the exact same steps as in the proof of Lemma~\ref{lemma_unique_limit},
	\[{\overline{\Omega_k^\varepsilon}}(n_k\mid X_k,X_k)=\frac{1}{N_k},\]
	where $N_k$ is the number of players in $\N_k$. Let $p^*$ denote the vector of strategies such that each player $i$ is assigned strategy $p_i^{X_k}$ if $i\in\N_k$. Then, for each player $i\in\N_k$, and all $k\in\{1,2,...,K\}$, we have
	\[u_i^\varepsilon(p_{-i}^*\mid X_k,s_i)=\frac{\int xg(x)\phi^\varepsilon(X_k-x)\wrt x}{\int g(x)\phi^\varepsilon(X_k-x)\wrt x}+w_i\left(N_1+N_2+\hdots +N_{k-1}+n\right)\Omega^\varepsilon_k(n\mid X_k,X_k)-c_i+s_i,\]
	which, choosing $\varepsilon$ sufficiently small, tends to:
	\begin{align*}
		u_i^\varepsilon(p_{-i}^*\mid X_k,s_i)&\to X_k+w_i\left(N_1+N_2+\hdots +n\right){\overline{\Omega^\varepsilon_k}}(n\mid X_k,X_k)-c_i+s_i\\
		&=X_k+\sum_{n=0}^{N_k-1}\frac{w_i(N_1+N_2+\hdots+N_{k-1}+n)}{N_k}-c_i+s_i.
	\end{align*}
	Setting $K=2$, $X_1=\tilde{x}_1$, $X_2=\tilde{x}_2$, and solving for $i\in\N_k$ the above for $u_i^\varepsilon(p_{-i}^*\mid X_k,\tilde{s}_i)=0$ yields the result.

\end{proof}

\noindent PROOF OF PROPOSITION~\ref{cor_joint}

\begin{proof}
	Observe that, in the notation of \eqref{eq_payoff}, we have $w_i(\sum_{j\neq i}a_j)=b_i$ for all $\sum_{j\neq i}a_j\geq I-1$ and $w_i(\sum_{j\neq i}a_j)=0$ for all $\sum_{j\neq i}a_j<I-1$. Therefore, the expected incentive to invest of player $i$ is
	\[u_i^\varepsilon(p_{-i}^{\tilde{x}}\mid \tilde{x},s_i)\to x_i^\varepsilon+\sum_{n=0}^{N-1}\frac{w_i(n)}{N}-c_i+s_i\]
	for $\varepsilon$ sufficiently small. Noting that $w_i(n)=b_i$ for all $n\geq I-1$ and 0 otherwise and that $N-n^*$ is the smallest integer greater than $I$, we obtain
	\[u_i^\varepsilon(p_{-i}^{\tilde{x}}\mid \tilde{x},\tilde{s}_i)\to \tilde{x}+\sum_{n=N-n^*-1}^{N-1}\frac{w_i(n)}{N}-c_i+\tilde{s}_i=\tilde{x}+\frac{n^*}{N}\cdot b_i-c_i+\tilde{s}_i=0.\]
	Solving for $\tilde{s}_i$ when $\tilde{x}=0$ yields the result.
\end{proof}

\bibliography{library}
\bibliographystyle{apa}

\end{document}